\documentclass{LMCS}

\usepackage{amsmath}
\usepackage{graphicx}
\usepackage{graphics}
\usepackage{epsfig}
\usepackage{amssymb}
\usepackage{fancybox}
\usepackage{verbatim}
\usepackage{listings}
\usepackage{color}
\usepackage{multicol}
\usepackage{rotating}
\usepackage[plain]{fancyref}
\usepackage[tight]{subfigure}
\usepackage{ntabbing}
\usepackage{url}            
\usepackage{hyperref,enumerate}

\newcommand*{\fancyrefdeflabelprefix}{def}
\newcommand*{\Frefdefname}{Definition}
\Frefformat{plain}{\fancyrefdeflabelprefix}{\Frefdefname~#1}

\newcommand{\COMMENT}[1]{}

\newcommand{\changed}[1]{{#1}}
\newcommand{\deleted}[1]{}

\newcommand{\ie}{i.e.}
\newcommand{\eg}{e.g.}

\newcounter{programlinenumber}

\newcommand{\Set}[1]{\{ \; {#1} \; \}}
\newcommand{\vbar}{\; | \;}

\newcommand{\squishlist}{
   \begin{list}{$\bullet$}
    { \setlength{\itemsep}{0pt}      \setlength{\parsep}{3pt}
      \setlength{\topsep}{3pt}       \setlength{\partopsep}{0pt}
      \setlength{\leftmargin}{1.5em} \setlength{\labelwidth}{1em}
      \setlength{\labelsep}{0.5em} } }

\newcommand{\squishlisttwo}{
   \begin{list}{$\bullet$}
    { \setlength{\itemsep}{0pt}    \setlength{\parsep}{0pt}
      \setlength{\topsep}{0pt}     \setlength{\partopsep}{0pt}
      \setlength{\leftmargin}{2em} \setlength{\labelwidth}{1.5em}
      \setlength{\labelsep}{0.5em} } }

\newcommand{\squishend}{
    \end{list}  }

\newcommand{\state}{\sigma}
\newcommand{\assertion}{\varphi}
\newcommand{\spec}{\Phi}
\newcommand{\acqpredSym}{\texttt{acquire}}
\newcommand{\relpredSym}{\texttt{release}}

\newcommand{\acqpred}[1]{\texttt{\acqpredSym}({#1})}
\newcommand{\relpred}[1]{\texttt{\relpredSym}({#1})}
\newcommand{\prf}{\mu}
\newcommand{\pred}{\varphi}
\newcommand{\predset}{\mathcal{R}}
\newcommand{\lockset}{\mathcal{L}}
\newcommand{\predmap}{\mathfrak{pm}}
\newcommand{\lockmap}{\mathfrak{lm}}

\newcommand{\edge}[2]{{#1} \rightarrow {#2}}
\newcommand{\stmtedge}[3]{{#1} \xrightarrow{#2} {#3}}
\newcommand{\assertedge}[3]{{#1} \xrightarrow{#2} {#3}}
\newcommand{\triple}[3]{\{{#1}\} {#2} \{{#3}\}}
\newcommand{\mustpreceedarrow}{\rightarrowtail}

\newcommand{\mustpreceedtcarrow}{\mustpreceedarrow^*}
\newcommand{\mustpreceedtc}[2]{{#1} \mustpreceedtcarrow {#2}}
\newcommand{\lockeq}{\rightleftarrows}
\newcommand{\seqsat}[2]{{#1} \models_s {#2}}
\newcommand{\concsat}[3]{{(#1,#2)} \models_c {#3}}
\newcommand{\cfg}[1]{\ensuremath{G_{#1}}}

\newcommand{\tredge}[1]{\xrightarrow{#1}}

\newcommand{\uniq}[1]{\hat{#1}}
\newcommand{\uniqv}[1]{\hat{#1}}

\newcommand{\kw}[1]{{\bf {\tt #1}}}
\newcommand{\kwb}[1]{\textcolor{blue}{#1}}
\newcommand{\lnrf}[1]{Line~\ref{lr:#1}}
\newcommand{\mypara}[1]{\vspace{0.6em}\noindent{\em #1.}}

\def\doi{7 (3:10) 2011}
\lmcsheading%
{\doi}
{1--29}
{}
{}
{Jun.26, 2010}
{Sep.~\phantom02, 2011}
{}

\begin{document}

\title{Logical Concurrency Control From Sequential Proofs}

\author[J.~Deshmukh]{Jyotirmoy Deshmukh\rsuper a}	
\address{{\lsuper a}University of Texas at Austin}	
\email{jyotirmoy@cerc.utexas.edu}  
\thanks{{\lsuper a}Work done while at Microsoft Research India}	

\author[G.~Ramalingam]{G.~Ramalingam\rsuper b}	
\address{{\lsuper{b,c,d}}Microsoft Research, India}	
\email{grama@microsoft.com, rvprasad@microsoft.com, kapilv@microsoft.com}

\author[V.-P.~Ranganath]{Venkatesh-Prasad Ranganath\rsuper c}	

\author[K.~Vaswani]{Kapil Vaswani\rsuper d}	


\keywords{concurrency control, program synthesis}
\subjclass{D.1.3, D.2.4, F.3.1}

\maketitle

\begin{abstract}
We are interested in identifying and enforcing the \emph{isolation requirements}
of a concurrent program, \ie, concurrency control that
ensures that the program meets its specification. The thesis of this paper
is that this can be done systematically starting from a sequential proof,
\ie, a proof of correctness of the program in the absence of concurrent
interleavings.
We illustrate our thesis by presenting a solution to the
problem of making a sequential library thread-safe for concurrent clients. 
We consider a sequential library annotated with assertions along with
a proof that these assertions hold in a sequential execution.
We show how we can use the proof to derive concurrency control 
that ensures that any execution of the library methods, when invoked
by concurrent clients, satisfies the same assertions.
We also present an extension to guarantee that the library methods are
linearizable or atomic.
\end{abstract}




\newcommand{\lockvar}{\textit{lock}}
\newcommand{\expOp}{\textit{f}}
\newcommand{\ProofS}{\ensuremath{\mathcal F}}
\newcommand{\Lib}{\ensuremath{\mathcal L}}
\newcommand{\SeqC}{\ensuremath{\mathcal P}}
\newcommand{\SpecS}{\ensuremath{\mathcal S}}
\newcommand{\ConcC}{\ensuremath{\mathcal{CP}}}
\newcommand{\SeqClient}{\ensuremath{\mathcal C}}
\newcommand{\ttatomic}{\texttt{atomic{}}}

\section{Introduction}\label{sec:intro}

A key challenge in concurrent programming is identifying and enforcing
the \emph{isolation requirements} of a program:
determining what constitutes undesirable interference between different threads
and implementing concurrency control mechanisms that prevent this.
In this paper, we show how a solution to this problem can be obtained
systematically from a \emph{sequential proof}: a proof that the program
satisfies a specification in the absence of concurrent interleaving. 

\mypara{Problem Setting} We illustrate our thesis by considering the
concrete problem of making a sequential library safe for concurrent clients.
Informally, given a sequential library that works correctly when
invoked by any sequential client, we show how to synthesize concurrency
control code for the library that ensures that it will work
correctly when invoked by any concurrent client.

\subsection*{Part I: Ensuring Assertions In Concurrent Executions}

Consider the example in \Fref{fig:motivating-example}(a).
The library consists of one procedure {\tt Compute}, which applies an
expensive function \expOp{} to an input variable {\tt num}.
As a performance optimization,
the implementation caches the last input and result.  If the
current input matches the last input, the last computed result is
returned.

\begin{figure}[h]
\centering
\begin{minipage}[h]{4.5in}
    \subfigure[][]{%
        \label{fig:ex1-a}
        \begin{minipage}[b]{0.1in}
        {\footnotesize {\tt
        \begin{ntabbing}
12\=12\=12\=123\=123\=123\=123\=123\=123\=123\=123\=123\kill
\kw{int} lastNum = 0;               \label{}                \\
\kw{int} lastRes = $\expOp$(0);     \label{}                \\
/* @returns \expOp(num) */   \label{lr:ex1-spec}     \\
\kw{Compute}(num) \{                \label{}                \\
\> /* \kwb{acquire (l);} */         \label{lr:ex1-conc1}    \\
\> \kw{if}(lastNum==num) \{         \label{lr:ex1-cond}     \\
\> \> res = lastRes;                \label{lr:ex1-next}     \\
\> \} \kw{else} \{                  \label{lr:ex1-else}     \\
\> \> /* \kwb{release (l);} */      \label{lr:ex1-conc2}    \\
\> \> res = $\expOp$(num);          \label{lr:ex1-ln1}      \\
\> \> /* \kwb{acquire (l);} */      \label{lr:ex1-conc3}    \\
\> \> lastNum = num;                \label{lr:ex1-ln2}      \\
\> \> lastRes = res;                \label{lr:ex1-overwrt}  \\
\> \}                               \label{}                \\
\> /* \kwb{release (l);} */         \label{lr:ex1-conc4}    \\
\> \kw{return} res;                 \label{}                \\
\}                                  \label{}                \\
        \end{ntabbing}}}
        \end{minipage}
    }
    \subfigure[][]{%
        \label{fig:ex1-b}
        \includegraphics[width=2.6in]{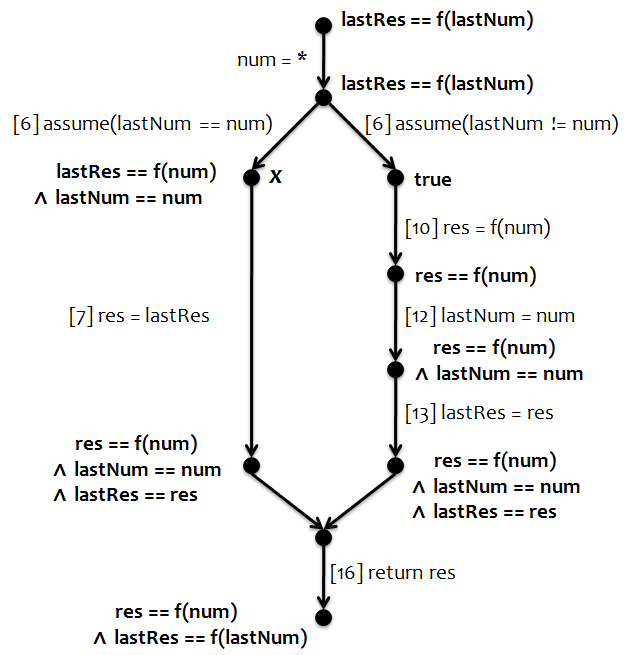}
    }%

\caption[Concurrency control for procedure {\tt Compute}]{
\subref{fig:ex1-a} Procedure {\tt Compute} (excluding
Lines~\ref{lr:ex1-conc1},\ref{lr:ex1-conc2},\ref{lr:ex1-conc3},\ref{lr:ex1-conc4})
applies a (side-effect free) function $\expOp$ to a parameter {\tt num} and
caches the result for later invocations.
Lines~\ref{lr:ex1-conc1},\ref{lr:ex1-conc2},\ref{lr:ex1-conc3},\ref{lr:ex1-conc4}
contain a lock-based concurrency control generated by our technique.
\subref{fig:ex1-b} The control-flow graph of {\tt Compute}, its edges
labeled by statements of {\tt Compute} and nodes labeled by proof assertions.
}%
\label{fig:motivating-example}
\end{minipage}
\end{figure} 
\noindent This procedure works correctly when used by a sequential client,
but not in the presence of concurrent procedure invocations.
E.g., consider an invocation of \texttt{Compute(5)} followed by 
concurrent invocations of \texttt{Compute(5)} and \texttt{Compute(7)}.
Assume that the second invocation of \texttt{Compute(5)} evaluates the
condition in \lnrf{ex1-cond}, and proceeds to \lnrf{ex1-next}. Assume
a context switch occurs at this point, and the invocation of
\texttt{Compute(7)} executes completely, overwriting \texttt{lastRes}
in \lnrf{ex1-overwrt}. Now, when the invocation of \texttt{Compute(5)}
resumes, it will erroneously return the (changed) value of
\texttt{lastRes}.  

In this paper, we present a technique that can detect the potential for such
interference and synthesize concurrency control to prevent the same.
The (lock-based) solution synthesized by our technique for the above example is
shown (as comments) in Lines 5, 9, 11, and 15 in \Fref{fig:ex1-a}.
With this concurrency control, the example works correctly even for concurrent
procedure invocations while permitting threads to perform the expensive
function $f$ concurrently. 

\mypara{The Formal Problem}
Formally, we assume that the correctness criterion for the library
is specified as a set of assertions and that 
the library satisfies these assertions in any execution of any sequential 
client.  \emph{ Our goal is to ensure that any execution of the library 
with any concurrent client also satisfies the given assertions.}
For our running example in \Fref{fig:ex1-a}, \lnrf{ex1-spec} specifies
the desired functionality for procedure \texttt{Compute}: {\tt
Compute} returns the value \expOp(\texttt{num}). 

\mypara{Logical Concurrency Control From Proofs}
A key challenge in coming up with concurrency control is determining
what interleavings between threads are safe.
A conservative solution may reduce concurrency by preventing correct interleavings.
An aggressive solution may enable more concurrency but introduce bugs.

The fundamental thesis we explore is the following: a proof that a
code fragment satisfies certain assertions in a sequential execution
precisely identifies the properties relied on by the code at different
points in execution; hence, such a sequential proof clearly identifies what
concurrent interference can be permitted; thus, a correct concurrency control
can be systematically (and even automatically) derived from such a proof.

We now provide an informal overview of our approach by illustrating
it for our running example.  \Fref{fig:ex1-b} presents a proof of
correctness for our running example (in a sequential setting).  The
program is presented as a control-flow graph, with its edges
representing program statements.  (The statement ``\texttt{num = *}''
at the entry edge indicates that the initial value of parameter
\texttt{num} is unknown.)
A proof consists of an invariant $\prf(u)$ attached
to every vertex $u$ in the control-flow graph (as illustrated in
the figure) such that:
(a) for every edge $\edge{u}{v}$ labelled with a statement
$s$, execution of $s$ in a state satisfying $\prf(u)$ is guaranteed
to produce a state satisfying $\prf(v)$, 
(b) The invariant $\prf(entry)$ attached to the entry vertex is satisfied
by the initial state and is implied by the invariant $\prf(exit)$ attached
to the exit vertex, and
(c) for every edge
$\edge{u}{v}$ annotated with an assertion $\assertion$, we have
$\prf(u) \Rightarrow \assertion$.
Condition (b) ensures that the proof is valid over any \emph{sequence of}
executions of the procedure.

The invariant $\prf(u)$ at vertex $u$ indicates the property required
(by the proof) to hold at $u$ to ensure that a sequential
execution satisfies all assertions of the library. We can reinterpret this
in a concurrent setting as follows: when a thread $t_1$ is at point $u$, it can tolerate
changes to the state by another thread $t_2$ as long as the invariant
$\prf(u)$ continues to hold from $t_1$'s perspective; however, if another
thread $t_2$ were to change the state such that $t_1$'s invariant
$\prf(u)$ is broken, then the continued execution by $t_1$ may
fail to satisfy the desired assertions.

Consider the proof in \Fref{fig:ex1-b}. The vertex labeled $x$ in the
figure corresponds to the point before the execution of
\lnrf{ex1-next}.  The invariant attached to $x$ indicates that the
proof of correctness depends on the condition $lastRes$==$f(num)$ being
true at $x$.  The execution of \lnrf{ex1-ln1} by another thread will
not invalidate this condition.  But, the execution of
\lnrf{ex1-overwrt} by another thread can potentially invalidate this
condition. Thus, we infer that, when one thread is at point $x$, an
execution of \lnrf{ex1-overwrt} by another thread should be avoided.

We prevent the execution of a statement $s$ by one thread when another thread
is at a program point $u$ (if $s$ might invalidate a predicate $p$ that is
required at $u$) as follows. We introduce a lock $\ell_p$ corresponding to
$p$, and ensure that every thread holds $\ell_p$ at $u$ and
ensure that every thread holds $\ell_p$ when executing $s$.

Our algorithm does this as follows.
From the invariant $\prf(u)$ at vertex $u$, we compute a set of predicates
$\predmap(u)$. (For now, think of $\prf(u)$ as the conjunction of predicates
in $\predmap(u)$.) $\predmap(u)$ represents the set of predicates required
at $u$. For any edge $\edge{u}{v}$, any predicate $p$ that is in
$\predmap(v) \setminus \predmap(u)$ is required at $v$ but not at $u$.
Hence, we acquire the lock for $p$ along this edge. Dually, for any predicate
that is required at $u$ but not at $v$, we release the lock along the edge.
As a special case, we acquire (release) the locks for all predicates in
$\predmap(u)$ at procedure entry (exit) when $u$ is the procedure entry (exit)
vertex.
Finally, if the execution of the statement on edge $\edge{u}{v}$ can invalidate
a predicate $p$ that is required at some vertex, we acquire and release
the corresponding lock before and after the statement (unless it is already
a required predicate at $u$ or $v$).
Note that our approach conservatively assumes that any two statements in
the library may be simultaneously executed by different threads. If an analysis
can identify that certain statements cannot be simultaneously executed (by
different threads), this information can be exploited to improve the solution,
but this is beyond the scope of this paper.

Our algorithm ensures that the locking scheme does not lead to
deadlocks by merging locks when necessary, as described later.
Finally, we optimize the synthesized solution using a few simple
techniques. E.g., in our example whenever the lock corresponding to
{\tt lastRes == res} is held, the lock for {\tt lastNum == num} is also held.
Hence, the first lock is redundant and can be eliminated.

\Fref{fig:motivating-example} shows the resulting library
with the concurrency control we synthesize. This implementation satisfies its
specification even in a concurrrent setting.  The synthesized solution permits
a high degree to concurrency since it allows multiple threads to compute 
\expOp{} concurrently.  A more conservative but correct locking
scheme would hold the lock during the entire procedure execution.

A distinguishing aspect of our algorithm is that it requires only
local reasoning and not reasoning about interleaved executions, as is common
with many analyses of concurrent programs.
Note that the synthesized solution depends on the proof used.
Different proofs can potentially yield different concurrency control solutions
(all correct, but with potentially different performance). 

We note that our approach has a close connection to the Owicki-Gries~\cite{owicki76verifying} approach
to computing proofs for concurrent programs. The Owicki-Gries approach shows
how the proofs for two statements can be composed into a proof for the
concurrent composition of the statements \emph{if the two statements do not
interfere with each other}. Our approach detects potential
interference between statements and inserts concurrency-control so that the
interference does not occur (permitting a safe concurrent composition of the
statements).

\mypara{Implementation} We have implemented our algorithm, using an existing
software model checker to generate the sequential proofs.
We used the tool to successfully synthesize concurrency control for several small examples.
The synthesized solutions are equivalent to those an expert programmer
would use.

\subsection*{Part II: Ensuring Linearizability}
The above approach can be used to ensure that concurrent executions
guarantee desired safety properties, preserve data-structure
invariants, and meet specifications (\eg, given as a precondition/postcondition pair).
Library implementors may, however, wish to provide the stronger guarantee of
linearizability with respect to the sequential specification:
\emph{
any concurrent execution of a procedure is guaranteed to satisfy its specification
and appears to take effect instantaneously at some point during its execution}.
In the second half of the paper, we show how the techniques sketched above can be extended
to guarantee linearizability.

\subsection*{Contributions}

We present a technique for synthesizing concurrency
control for a library (\eg, developed for use by a single-threaded client)
to make it safe for use by concurrent clients.
However, we believe that the key idea we present -- a technique for identifying and
realizing isolation requirements from a sequential proof -- can be used in
other contexts as well (\eg, in the context of a whole program consisting of
multiple threads, each with its own assertions and sequential proofs).

Sometimes a library designer may choose to delegate the responsibility for
concurrency control to the clients of the library and not make the library
thread-safe\footnote{
This may be a valid design option in some cases.
However, in examples such as our running example, this could be a bad idea.
}.
Alternatively, library implementers could choose to make the execution
of a library method appear atomic by wrapping it in a transaction and executing it in
an STM (assuming this is feasible). 
These are valid options but orthogonal to the point of this paper.
Typically, a program is a software stack, with each level serving as
a library. Passing the buck, with regards to
concurrency control, has to stop somewhere. Somewhere in the stack,
the developer needs to decide what degree of isolation is required by
the program; otherwise, we would end up with a program consisting of
multiple threads where we require every thread's execution to appear
atomic, which could be rather severe and restrict concurrency
needlessly.
The ideas in this paper provide a systematic method for determining
the isolation requirements. While we illustrate the idea in a
simplified setting, it should ideally be used at the appropriate level
of the software stack.

In practice, full specifications are rarely available.
We believe that our technique can be used even with
lightweight specifications or in the absence of specifications.
Consider our example in Fig.~\ref{fig:motivating-example}.
A symbolic analysis of this library, with a harness
representing a sequential client making an arbitrary sequence of calls
to the library, can, in principle, infer that the returned
value equals \texttt{f(num)}. As the returned value is the only observable
value, this is the strongest functional specification a user can write.
Our tool can be used with such an inferred specification as well. 

\mypara{Logical interference}
Existing concurrency control mechanisms (both pessimistic as well as optimistic)
rely on a data-access based
notion of interference: concurrent accesses to the same data, where at
least one access is a write, is conservatively treated as interfence.
A contribution of this paper is that it introduces a more logical/semantic notion of
interference that can be used to achieve more permissive,
yet safe, concurrency control. Specifically, concurrency control based
on this approach permits interleavings that existing schemes based on
stricter notion of interference will disallow.
Hand-crafted concurrent code often permits ``benign interference'' 
for performance reasons,
suggesting that programmers do rely on such a logical notion of interference.

\newcommand{\ttt}{\texttt}
\newcommand{\lib}{\ensuremath{\mathcal L}}
\newcommand{\libcc}{\ensuremath{\hat{\lib}}}
\newcommand{\ProcSet}{\ensuremath{\mathcal P}}
\newcommand{\Proc}{\textit{P}}
\newcommand{\Locks}{Lk}
\newcommand{\TidSet}{T}
\newcommand{\stmt}{S}
\newcommand{\GV}{\ensuremath{V_G}}
\newcommand{\gv}{\textit{sv}}
\newcommand{\LVProc}{\ensuremath{V_\Proc}}
\newcommand{\LV}{\ensuremath{V_L}}
\newcommand{\lv}{\textit{lv}}
\newcommand{\lock}{\textit{lk}}
\newcommand{\PC}{\textit{pc}}
\newcommand{\entryVertex}[1]{\ensuremath{N_{#1}}}
\newcommand{\exitVertex}[1]{\ensuremath{X_{#1}}}
\newcommand{\qvertex}{\ensuremath{w}}
\newcommand{\instate}{\state_{in}}
\newcommand{\outstate}{\state_{out}}
\newcommand{\lockstate}{\state_{lk}}
\newcommand{\idlestate}{\state^\ell_{Q}}
\newcommand{\slStates}{\Sigma^s}
\newcommand{\sllStates}{\Sigma^s_\ell}
\newcommand{\slgStates}{\Sigma^s_g}
\newcommand{\clStates}{\Sigma^c}
\newcommand{\cllStates}{\Sigma^c_\ell}
\newcommand{\clgStates}{\Sigma^c_g}
\newcommand{\lockStates}{\Sigma^c_\text{lk}}
\newcommand{\threadIds}{T}
\newcommand{\lzty}{Linearizability}

\newcommand{\sltr}{\rightsquigarrow_{\textit{sl}}}
\newcommand{\sctr}{\rightsquigarrow_{\textit{sc}}}
\newcommand{\str}{\rightsquigarrow_{\textit{s}}}
\newcommand{\ctr}{\rightsquigarrow_{\textit{c}}}
\newcommand{\lsarrow}[3]{{\stackrel{#1}{#2}_{#3}}}
\newcommand{\sltrOf}[1]{\lsarrow{#1}{\rightsquigarrow}{\textit{sl}}}
\newcommand{\sctrOf}[1]{\lsarrow{#1}{\rightsquigarrow}{\textit{sc}}}
\newcommand{\strOf}[1]{\lsarrow{#1}{\rightsquigarrow}{\textit{s}}}
\newcommand{\ctrOf}[1]{\lsarrow{#1}{\rightsquigarrow}{\textit{c}}}
\newcommand{\longlsarrow}[3]{{\stackrel{#1}{{#2}_{#3}}}}
\newcommand{\longstrOf}[1]{\longlsarrow{#1}{\rightsquigarrow}{\textit{s}}}
\newcommand{\longctrOf}[1]{\longlsarrow{#1}{\rightsquigarrow}{\textit{c}}}

\newcommand{\exec}{\pi}

\section{The Problem}\label{sec:problem}

In this section, we introduce required terminology and formally define the problem.
Rather than restrict ourselves to a specific syntax for programs and assertions,
we will treat them abstractly, assuming only that they can be given a semantics
as indicated below, which is fairly standard.

\subsection{The Sequential Setting}

\noindent{\em Sequential Libraries.} %
A library $\lib$ is a pair $(\ProcSet, \GV)$, where $\ProcSet$ is a
set of procedures (defined below), and $\GV$ 
is a set of variables, termed \emph{global}
variables, accessible to all and only procedures in $\ProcSet$.
A procedure $\Proc \in \ProcSet$ is a pair $(\cfg{\Proc}, \LVProc)$, where $\cfg{\Proc}$ is a
control-flow graph with each edge labeled by a primitive statement, and $\LVProc$
is a set of variables, referred to as \emph{local} variables,
restricted to the scope of $\Proc$.
(Note that $\LVProc$ includes the formal parameters of $\Proc$ as well.)
To simplify the semantics, we will
assume that the set $\LVProc$ is the same for all procedures and denote
it $\LV$.  

Every control-flow graph has a unique entry vertex
$\entryVertex{\Proc}$ (with no predecessors) and a unique
exit vertex $\exitVertex{\Proc}$ (with no successors).
Primitive statements are either \texttt{skip} statements, assignment
statements, {\tt assume} statements, {\tt return} statements, or
{\tt assert} statements.
An {\tt assume} statement is used to implement conditional control flow as usual.
Given control-flow graph nodes $u$ and $v$, we denote an edge from $u$ to
$v$, labeled with a primitive statement $S$, by $\stmtedge{u}{S}{v}$.

To reason about all possible sequences of invocations of the library's
procedures,
we define the \emph{control graph} of a library to be 
the union of the control-flow graphs of all the procedures, augmented by
a new vertex \qvertex{}, as well as an edge from every procedure exit vertex to
\qvertex{} and an edge from \qvertex{} to every procedure entry vertex.  We
refer to \qvertex{} as the \emph{quiescent} vertex.  Note that a one-to-one
correspondence exists between a path in the control graph of the library,
starting from \qvertex{}, and the execution of a sequence of procedure calls.
The edge $\edge{\qvertex}{\entryVertex{\Proc}}$ from the quiescent vertex to
the entry vertex of a procedure $\Proc$ models an arbitrary call to procedure
$\Proc$.
We refer to these as \emph{call edges}. 

\mypara{Sequential States} %
A procedure-local state $\state_\ell \in \sllStates$ is a pair $(\PC, \state_d)$
where $\PC$, the program counter, is a vertex in the control graph and
$\state_d$ is a map from the local variables $\LV$ to their values.
(We use ``$s$'' as a superscript or subscript to indicate elements of the
semantics of sequential execution.)
A global state $\state_g \in \slgStates$ is a map from global variables $\GV$ to their values.
A library state $\state$ is a pair $(\state_\ell, \state_g) \in \sllStates \times \slgStates$.
We define $\slStates$ to be $\sllStates \times \slgStates$.
We say that a state is a \emph{quiescent state} if its \PC{} value
is \qvertex{} and that it is an \emph{entry state} if its \PC{}
value equals the entry vertex of some procedure. 

\mypara{Sequential Executions} %
We assume a standard semantics for primitive statements that can be
captured as a transition relation $\str \;
\subseteq \slStates \times \slStates$ as follows.  Every control-flow edge $e$
induces a transition relation $\strOf{e}$, where $\state \strOf{e} \state'$ iff
$\state'$ is one of the possible outcomes of the execution of (the statement labeling)
the edge $e$ in state $\state$.
The edge $\edge{\qvertex}{\entryVertex{\Proc}}$ from the
quiescent vertex to the entry vertex of a procedure $\Proc$ models an arbitrary
call to procedure $\Proc$.  Hence, in defining the transition relation, such
edges are treated as statements that assign a non-deterministically chosen
value to every formal parameter of $\Proc$ and the default initial value to
every local variable of $\Proc$.  Similarly, the edge
$\edge{\exitVertex{\Proc}}{\qvertex}$ is treated as a \texttt{skip}
statement.
We say $\state \str \state'$ if there exists some edge $e$ such that $\state
\strOf{e} \state'$.

A \emph{sequential execution} is a sequence of states $\state_0
\state_1 \cdots \state_k$
where $\state_0$ is the initial state of the library and we have $\state_i \str
\state_{i+1}$ for $0 \leq i < k$.
A sequential execution represents the execution of a sequence of calls
to the library's procedures (where the last call's execution may
be incomplete).  Given a sequential execution $\state_0 \state_1 \cdots
\state_k$, we say that $\state_i$ is the \emph{corresponding entry state} of
$\state_j$ if $\state_i$ is an entry state and no state $\state_h$ is an entry
state for $i < h \leq j$.

\mypara{Sequential Assertions}%
\changed{
We use {\tt assert} statements to specify desired correctness properties
of the library. \texttt{Assert} statements have no effect on the execution
semantics and are equivalent to {\tt skip} statements in the semantics.
Assertions are used only to define the notion of \emph{well-behaved}
executions as follows.

An {\tt assert} statement is of the form {\tt assert} $\theta$ where,
$\theta$ is a 1-state assertion $\assertion$ or a 2-state assertion $\spec$.
A 1-state assertion, which we also refer to as a predicate, makes an
assertion about a single library state. Rather than define a specific
syntax for assertions, we assume that the semantics of assertions are
defined by a relation $\seqsat{\state}{\assertion}$ denoting that a
state $\state$ satisfies the assertion $\assertion$.

1-state assertions can be used to specify the invariants expected at
certain program points.
In general, specifications for procedures take the form of
two-state assertions, which relate the input state to output state. 
We use 2-state assertions for this purpose.
The semantics of a 2-state assertion $\spec$ is assumed to be defined
by a relation
$\seqsat{(\instate, \outstate)}{\spec}$ (meaning that state $\outstate$
satisfies assertion $\spec$ with respect to state $\instate$).
In our examples, we use special input variables $v^{in}$ to refer to the
value of the variable $v$ in  the first state.
E.g., the specification ``$x == x^{in}+1$'' asserts that the value of
$x$ in the second state is one more than its value in the first
state.
}

\begin{defi}
  \label{def:seq-spec-correctness}
  \changed{
  A sequential execution is said to satisfy the library's assertions if for
  any transition
  $\state_i \longstrOf{e} \state_{i+1}$ in the execution,
  where $e$ is labelled by the statement ``$\texttt{assert } \theta$'',
  we have
  (a) $\seqsat{\state_i}{\theta}$ if $\theta$ is a 1-state assertion, and
  (b) $\seqsat{(\instate,\state_i)}{\theta}$ where $\instate$ is the
  corresponding entry state of $\state_i$, otherwise.
  A sequential library satisfies its
  specifications if every execution of the library satisfies its specifications.
  }
\end{defi}

\subsection{The Concurrent Setting}

\noindent{\em Concurrent Libraries.} %
A concurrent library $\lib$ is a triple $(\ProcSet, \GV, \Locks)$, where
$\ProcSet$ is a set of concurrent procedures, $\GV$ is a set of global 
variables, and $\Locks$ is a set of global locks.  A concurrent procedure is like a
sequential procedure, with the extension that a primitive statement is either
a sequential primitive statement or a locking statement 
of the form $\acqpred{\ell}$ or $\relpred{\ell}$ where $\ell$ is a lock.

\mypara{Concurrent States} %
A concurrent library permits concurrent invocations of
procedures. We associate each procedure invocation with a thread
(representing the client thread that invoked the procedure).  Let $\TidSet$
denote an infinite set of thread-ids, which are used as unique identifiers for
threads.
In a concurrent execution, every thread has a private copy of local variables,
but all threads share a single copy of the global variables.  Hence, the
local-state in a concurrent execution is represented by a map from $\TidSet$ to
$\sllStates$.  (A thread whose local-state's \PC{} value is the quiescent point
represents an idle thread, \ie, a thread not processing any procedure
invocation.)
Let $\cllStates = \TidSet \rightarrow \sllStates$
denote the set of all local states.

At any point during execution, a lock $\lock$ is either free or held by one 
thread.  We represent the state of locks by a partial function
from $\Locks$ to $\TidSet$ indicating which thread, if any, holds any given lock.
(A lock that is not held by any thread will not be in the domain of the partial function.)
Let $\lockStates = \Locks \hookrightarrow \TidSet$ represent the set of
all lock-states.
Let $\clgStates$ = $\slgStates \times \lockStates$ denote the set of all global states.
Let $\clStates = \cllStates \times \clgStates$ denote the set of all states.
Given a concurrent state $\state = (\state_\ell,(\state_g,\lockstate))$ and thread $t$,
we define $\state[t]$ to be the sequential state $(\state_\ell(t),\state_g)$. 

\mypara{Concurrent Executions} %
The concurrent semantics is induced by the sequential semantics
as follows.
Let $e$ be any control-flow edge labelled with
a sequential primitive statement, and $t$ be any thread. We say that
$(\state_\ell,(\state_g,\lockstate)) \ctrOf{(t,e)}$
$(\state'_\ell,(\state'_g,\lockstate))$ iff $(\state_t,\state_g)
\strOf{e} (\state'_t,\state'_g)$ where $\state_t = \state_\ell(t)$ and
$\state'_\ell = \state_\ell[t \mapsto \state'_t]$.  The transitions
corresponding to lock acquire/release are defined in the obvious way.
We say that $\state \ctr \state'$ iff there exists some $(t,e)$ such that
$\state \ctrOf{(t,e)} \state'$.

A \emph{concurrent execution} is a sequence
$\state_0 \state_1 \cdots \state_k$, where $\state_0$ is the initial state of
the library and $\state_i \ctrOf{\ell_i} \state_{i+1}$ for $0 \leq i < k$,
where the label $\ell_i = (t_i,e_i)$ identifies the executing thread and
executed edge.
We say that $\ell_0 \cdots \ell_{k-1}$ is the \emph{schedule} of this execution.
A sequence $\ell_0 \cdots \ell_m$ is a
\emph{feasible schedule} if it is the schedule of some concurrent execution.
Consider a concurrent execution $\state_0 \state_1 \cdots \state_k$.  We say
that a state $\state_i$ is a $t$-entry-state if it is generated from a
quiescent state by thread $t$ executing a call edge.  We say that $\state_i$ is
the \emph{corresponding t-entry state} of $\state_j$ if $\state_i$ is a
$t$-entry-state and no state $\state_h$ is a $t$-entry-state for $i < h \leq j$.

We note that our semantics uses sequential consistency. Extending our
results to support weaker memory models is future work.

\mypara{Interpreting Assertions In Concurrent Executions} %
\changed{
In a concurrent setting, assertions are evaluated in
the context of the thread that executes the corresponding {\tt assert}
statement.  We say that state $\state$ satisfies a 1-state assertion
$\assertion$ in the context of thread $t_i$ (denoted by
$\concsat{\state}{t_i}{\assertion}$) iff $\seqsat{\state[t_i]}{\assertion}$.
For any 2-state assertion $\spec$, we say that a given pair of states $(\instate,
\outstate)$ satisfies $\spec$ in the context of thread $t$ (denoted by
$\concsat{(\instate,\outstate)}{t}{\,\spec}$) iff
$\seqsat{(\instate[t],\outstate[t])}{\spec}$.
}

\begin{defi}
  \label{def:conc-spec-correctness}
  \changed{
  A concurrent execution $\exec$ is said to satisfy an assertion 
  ``$\texttt{assert } \theta$'' labelling an edge $e$ if for
  any transition $\state_i \longctrOf{(t,e)} \state_{i+1}$ in the
  execution,
  we have
  (a) $\concsat{\state_i}{t}{\theta}$, if $\theta$ is a 1-state assertion, and
  (b) $\concsat{(\instate,\state_i)}{t}{\theta}$ where $\instate$
  is the corresponding $t$-entry state of $\state_i$, otherwise.
 The execution is said to satisfy the library's specification if it
 satisfies all assertions in the library.
    A concurrent library satisfies its specification if every execution of the
  library satisfies its specification.
  }
\end{defi}

\noindent{\em Frame Conditions.} %
Consider a library with two global variables \ttt{x} and \ttt{y} and
a procedure \ttt{IncX} that increments \ttt{x} by 1. A possible specification
for \ttt{IncX} is
${(x == x^{in} + 1)~\&\&~(y == y^{in})}$. The condition ${y == y^{in}}$ is \ttt{IncX}'s
frame condition, which says that it will not modify \ttt{y}.
Explicitly stating such frame conditions is unnecessarily restrictive, as
a concurrent update to \ttt{y} by another procedure, when \ttt{IncX} is
executing, would be considered a violation of \ttt{IncX}'s specification.
Frame conditions can be handled better by treating a specification as a pair $(S,\spec)$
where $S$ is the set of all global variables referenced by
the procedure, and $\spec$ is a specification that
does not refer to any global variables outside $S$. For our above example,
the specification will be (\{$x$\}, ${x == x^{in} + 1)}$).
In the sequel, however, we will restrict ourselves to the simpler setting
and ignore this issue.

\subsection{Goals}
Our goal is:
Given a sequential library \lib{} with assertions satisfied
in every sequential execution,
construct \libcc{}, by augmenting \lib{} with concurrency
control, such that every concurrent execution of \libcc{}
satisfies all assertions.
In Section~\ref{sec:linearizability}, we extend this goal to construct
\libcc{} such that every concurrent execution of \libcc{} is linearizable.

\newcommand{\obliged}{\mathfrak{om}}
\newcommand{\bothmap}{\mathfrak{m}}
\newcommand{\mayFalsify}{\mathfrak{mbf}}
\newcommand{\PLocksAcquired}{\textit{BasisLocksAcq}}
\newcommand{\PLocksReleased}{\textit{BasisLocksRel}}
\newcommand{\BreakLocks}{\textit{BreakLocks}}

\section{Preserving Single-State Assertions}
\label{sec:pres_assn}

In this section we describe our algorithm for synthesizing
concurrency control, but restrict our attention to single-state assertions.

\subsection{Algorithm Overview}
\label{sec:algo-overview}
A \emph{sequential proof} is a mapping $\prf$ from vertices of the control graph
to predicates such that
(a) for every edge $e = \stmtedge{u}{t}{v}$,
$\triple{\prf(u)}{t}{\prf(v)}$ is a valid Hoare triple
$\seqsat{\state_1}{\prf(u)}$
and
$\state_1 \strOf{e} \state_2$
implies
$\seqsat{\state_2}{\prf(v)}$),
and (b) for every edge
$\assertedge{u}{\textit{\texttt{assert }} \assertion}{v}$, we have $\prf(u)
\Rightarrow \assertion$.
Note that condition (a) requires $\triple{\prf(u)}{t}{\prf(v)}$ to be partially
correct. The execution of statement $t$ in a state satisfying $\prf(u)$
does not have to succeed. This is required primarily for the case when $t$
represents an \texttt{assume} statement. If we want to ensure that the
execution of a statement $t$ does not cause any runtime error, we can simply
replace $t$ by ``\texttt{assert $p$; $t$}'' where $p$ is the condition required to
ensure that $t$ does not cause any runtime error.)

Note that the invariant $\prf(u)$ attached to a vertex $u$ by a
proof indicates two things:
(i) any sequential execution reaching point $u$ will produce a
state satisfying $\prf(u)$, and
(ii) any sequential execution from point $u$, starting from a state
satisfying $\prf(u)$ will satisfy the invariants labelling other
program points (and satisfy all assertions encountered during the
execution).

A procedure that satisfies its assertions in a sequential execution
may fail to do so in a concurrent execution due to interference
by other threads.
E.g., consider a thread $t_1$ that reaches a program point
$u$ in a state that satisfies $\prf(u)$.
At this point, another thread $t_2$ may execute some statement that
changes the state to one where $\prf(u)$ no longer holds. Now, we no
longer have a guarantee that a continued execution by $t_1$ will
successfully satisfy its assertions.
The preceding paragraph, however, hints at the interference
we must avoid to ensure correctness:
when a thread $t_1$ is at point $u$, we should ensure that no
other thread $t_2$ changes the state to one where $t_1$'s invariant
$\prf(u)$ fails to hold. Any change to the state by another
thread $t_2$ can be tolerated by $t_1$ \emph{as long as the invariant
$\prf(u)$ continues to hold}.
We can achieve this by associating a lock with the invariant $\prf(u)$,
ensuring that $t_1$ holds this lock when it is at program point $u$,
and ensuring that any thread $t_2$ acquires this lock before executing
a statement that may break this invariant.
An invariant $\prf(u)$, in general, may be a boolean formula over simpler
predicates. We could potentially get different locking solutions by
associating different locks with different sub-formulae of the invariant.
We elaborate on this idea below.

A \emph{predicate mapping} is a mapping $\predmap$ from the vertices of
the control graph to a set of predicates.
A predicate mapping $\predmap$ is said to be a \emph{basis} for
a proof $\prf$ if every $\prf(u)$ can be expressed as a boolean formula
(involving conjunctions, disjunctions, and negation) over $\predmap(u)$.  A
basis $\predmap$ for proof $\prf$ is \emph{positive} if every $\prf(u)$ can be
expressed as a boolean formula involving only conjunctions and disjunctions
over $\predmap(u)$.

Given a proof $\prf$, we say that an edge $\stmtedge{u}{s}{v}$
\emph{sequentially positively preserves} a predicate $\pred$ if
$\triple{\prf(u) \land \pred}{s}{\pred}$ is a valid Hoare triple.
Otherwise, we say that the edge \emph{may sequentially falsify}
the predicate $\pred$.
Note that the above definition is in terms of the Hoare logic for our
sequential language.
However, we want to formalize the notion of a thread $t_2$'s execution
of an edge falsifying a predicate $\pred$ in a thread $t_1$'s scope.
Given a predicate $\pred$, let $\uniq{\pred}$ denote the predicate
obtained by replacing every local variable $x$ with a new unique
variable $\uniqv{x}$. We say that an edge $\stmtedge{u}{s}{v}$
\emph{may falsify} $\pred$ iff the edge may sequentially falsify
$\uniq{\pred}$.
(Note that this reasoning requires working with formulas with
free variables, such as $\uniqv{x}$. This is straightforward
as these can be handled just like extra program variables.)

E.g., consider \lnrf{ex1-overwrt} in Fig.~\ref{fig:motivating-example}.
Consider predicate $lastRes$==$f(num)$. By renaming local variable
{\tt num} to avoid naming conflicts, we obtain predicate {\tt lastRes
== f($\uniqv{{\mathtt{num}}}$)}.  We say that \lnrf{ex1-overwrt} {\em
may falsify} this predicate because the triple $\{res$ == $f(num)$
$\wedge$ $lastNum$ == $num$ $\wedge$ $lastRes$ == $f(\uniqv{num})\}$
{\tt lastRes = res} $\{lastRes$ == $f(\uniqv{num})\}$ is not a valid
Hoare triple.

Let $\predmap$ be a positive basis for a proof $\prf$ and $\predset = \cup_u
\predmap(u)$.
For any program point $u$, if a predicate $\pred$ is in $\predmap(u)$, we say that $\pred$ is \emph{relevant} at
program point $u$.  In a concurrent execution, we say that a predicate
$\pred$ is relevant to a thread $t$ in a given state if $t$ is at a program
point $u$ in the given state and $\pred \in \predmap(u)$.  Our locking scheme
associates a lock with every predicate $\pred$ in $\predset$.  The invariant it
establishes is that a thread, in any state, will hold the locks corresponding
to precisely the predicates that are relevant to it.  We will simplify the
initial description of our algorithm by assuming that distinct predicates are
associated with distinct locks and later relax this requirement.

Consider any control-flow edge $e = \stmtedge{u}{s}{v}$.
Consider any predicate $\pred$ in $\predmap(v) \setminus \predmap(u)$.
We say that predicate $\pred$ \emph{becomes relevant}\footnote{
Frequently it will be the case that the execution of statement $s$
makes predicate $\pred$ true. This is true if every invariant
$\prf(v)$ is a conjunction of the basis predicates in
$\predmap(u)$. Since we allow disjunctions as well, this is not, however,
always true.}
at edge $e$. In the motivating example, the predicate {\tt
lastNum == num} becomes relevant at \lnrf{ex1-ln2}

We ensure the desired invariant by acquiring the locks corresponding
to every predicate that becomes relevant at edge $e$
{\it prior to statement $s$ in the edge}.  (Acquiring the lock after $s$ may be too
late, as some other thread could intervene between $s$ and the acquire and
falsify predicate $\pred$.)

Now consider any predicate $\pred$ in $\predmap(u) \setminus \predmap(v)$.
We say that $\pred$ \emph{becomes irrelevant} at edge $e$.
E.g., predicate {\tt lastRes == f(lastNum)} becomes
irrelevant once the false branch at \lnrf{ex1-else} is taken.
For every $p$ that becomes irrelevant at edge $e$, we release the lock
corresponding to $p$ \emph{after} statement $s$.

The above steps ensure that in a concurrent execution a thread will hold
a lock on all predicates relevant to it.
The second component of the concurrency control mechanism is to ensure
that any thread acquires a lock on a predicate before it falsifies
that predicate.
Consider an edge $e = \stmtedge{u}{s}{v}$ in the control-flow graph.
Consider any predicate $\pred \in \predset$ that may be falsified by
edge $e$. We add an acquire of the lock corrresponding to this
predicate before $s$ (unless $\pred \in \predmap(u)$), and add a
release of the same lock after $s$ (unless $\pred \in \predmap(v)$).

\mypara{Managing locks at procedure entry/exit} %
We will need to acquire/release locks at procedure entry and exit
differently from the scheme above.
Our algorithm works with the control graph defined in \Fref{sec:problem}.
Recall that we use a quiescent vertex $\qvertex$ in the control graph.
The invariant $\prf(\qvertex)$ attached to this quiescent vertex
describes invariants maintained by the library (in between procedure
calls).
Any \texttt{return} edge $\stmtedge{u}{\tt return}{v}$ must be augmented to release
all locks corresponding to predicates in $\predmap(u)$ before returning.
Dually, any procedure entry edge $\edge{\qvertex}{u}$ must be augmented to
acquire all locks corresponding to predicates in $\predmap(u)$.

However, this is not enough. Let $\edge{\qvertex}{u}$ be a procedure $p$'s
entry edge. The invariant $\prf(u)$ is part of the library invariant
that procedure $p$ depends upon. It is important to ensure that when
a thread executes the entry edge of $p$ (and acquires locks corresponding to the
basis of $\prf(u)$) the invariant $\prf(u)$ holds. We achieve this by
ensuring that any procedure that invalidates the invariant $\prf(u)$
holds the locks on the corresponding basis predicates until it
reestablishes $\prf(u)$. We now describe how this can be done in a
simplified setting where the invariant $\prf(u)$ can be expressed
as the conjunction of the predicates in the basis $\predmap(u)$
for every procedure entry vertex $u$. (Disjunction can be handled at
the cost of extra notational complexity.) We will refer to the
predicates that occur in the basis $\predmap(u)$ of some procedure
entry vertex $u$ as \emph{library invariant predicates}.

We use an \emph{obligation} mapping $\obliged(v)$ that maps each vertex
$v$ to a set of library invariant predicates to track the invariant predicates
that may be invalid at $v$ and need to be reestablished before the procedure
exit. We say a function $\obliged$ is a valid obligation mapping if it
satisfies the following constraints for any edge $e = \edge{u}{v}$:
(a) if $e$ may falsify a library invariant $\pred$, then $\pred$ must
be in $\obliged(v)$, and
(b) if $\pred \in \obliged(u)$, then $\pred$ must be in $\obliged(v)$
unless $e$ \emph{establishes} $\pred$.
Here, we say that an edge $\stmtedge{u}{s}{v}$
\emph{establishes} a predicate $\pred$ if
$\triple{\prf(u)}{s}{\pred}$ is a valid Hoare triple.
Define $\bothmap(u)$ to be $\predmap(u) \cup \obliged(u)$. Now, the
scheme described earlier can be used, except that we use $\bothmap$
in place of $\predmap$.

\mypara{Locking along \texttt{assume} edges} %
Recall that we model conditional branching, based on a condition
$p$, using two edges labelled ``\texttt{assume} $p$'' and ``\texttt{assume} $!p$''.
Any lock to be acquired along an \texttt{assume} edge will need to be acquired
before the condition is evaluated.
If the lock is required along both \texttt{assume} edges, this is sufficient.
If the lock is not required along all \texttt{assume} edges out of a vertex,
then we will have to release the lock along the edges where it is not required.

\mypara{Deadlock Prevention} %
The locking scheme synthesized above may potentially lead to a deadlock.
We now show how to modify the locking scheme to avoid this possibility.
For any edge $e$, let $\mayFalsify(e)$ be (a conservative approximation of)
the set of all predicates that may be falsified by the execution of edge $e$.
We first define a binary relation $\mustpreceedarrow$ on the predicates used
(\ie, the set $\predset$) as follows:
we say
that $p \mustpreceedarrow r$ iff there exists a control-flow edge
$\stmtedge{u}{s}{v}$ such that $p \in \bothmap(u) \wedge r \in (\bothmap(v)
\cup \mayFalsify(\stmtedge{u}{s}{v})) \setminus \bothmap(u)$.
Note that $p \mustpreceedarrow r$
holds iff it is possible for some thread to try to acquire a lock on $r$ while
it holds a lock on $p$.  Let $\mustpreceedtcarrow$ denote the transitive
closure of $\mustpreceedarrow$.

We define an equivalence relation $\lockeq$ on $\predset$ as follows:
$p \lockeq r$ iff $\mustpreceedtc{p}{r} \wedge \mustpreceedtc{r}{p}$.
Note that any possible deadlock must involve an equivalence class of
this relation. We map all predicates in an equivalence class to the
same lock to avoid deadlocks.
In addition to the above, we establish a total ordering on all the locks,
and ensure that all lock acquisitions we add to a single edge are done
in an order consistent with the established ordering. (Note that the ordering
on the locks does not have to be total; a partial ordering is fine, as long
as any two locks acquired along a single edge are ordered.)  

\mypara{Improving The Solution} %
Our scheme can sometimes introduce \emph{redundant} locking.
E.g., assume that in the generated solution a
lock $\ell_1$ is always held whenever a lock $\ell_2$ is acquired. Then,
the lock $\ell_2$ is redundant and can be eliminated.
Similarly, if we have a predicate $\pred$ that is never falsified by any
statement in the library, then we do not need to acquire a lock for this
predicate. We can eliminate
such redundant locks as a final optimization pass over the generated solution.

\mypara{Using Reader-Writer Locks}
Note that a lock may be acquired on a predicate $\pred$ for one of two reasons
in the above scheme: either to ``preserve'' $\pred$ or to ``break'' $\pred$.
These are similar to read-locks and write-locks. 
Note that it is safe for multiple threads to simultaneously hold a
lock on the same predicate $\pred$ if they want to ``preserve'' it,
but a thread that wants to ``break'' $\pred$ needs an exclusive
lock. Thus, reader-writer locks can be used to improve concurrency,
but space constraints prevent a discussion of this extension.
However, since it is unsafe for a thread that
holds a read-lock on a predicate $\pred$ to try to acquire a
write-lock $\pred$, using this optimization also requires an extension
to the basic deadlock avoidance scheme.

Specifically, it is unsafe for a thread that holds a read-lock on a
predicate $\pred$ to try to acquire a write-lock $\pred$, as this can
lead to a deadlock. Hence, any acquisition of a lock on a predicate
$\pred$ (to preserve it) should be made an exclusive (write) lock if
along some execution path it may be necessary to promote this lock to
a write lock before the lock is released.


\mypara{Generating Proofs} %
The sequential proof required by our scheme can be generated using
verification tools such as SLAM~\cite{BallRajamaniSpin00},
BLAST~\cite{henzinger02lazy,BLAST:POPL04} or Yogi~\cite{gulvani06synergy}.
Predicate abstraction~\cite{BallRajamaniSpin00} is a program analysis
technique that constructs a conservative, finite state abstraction of
a program with a large (possibly infinite) state space using a set of
predicates over program variables. 
Tools such as SLAM and BLAST use predicate abstraction to check if a
given program $P$ satisfies a specification $\phi$.  The tools
start with a simple initial abstraction and iteratively refine
the abstraction until the abstraction is rich enough to prove the
absence of a concrete path from the program's initial state to an
error state (or a real error is identified).

For programs for which verification succeeds, the final abstraction
produced, as well as the result of abstract interpretation using this
abstraction, serve as a good starting point for constructing the
desired proof.
The final abstraction consists of a predicate map $\predmap$ which maps
each program point to a set of predicates and as well as a mapping
from each program statement to a set of abstract predicate transformers
which together define an abstract transition system.
Furthermore, abstract interpretation utilizing this abstraction
effectively computes a formula $\prf(u)$ over the set of predicates
$\predmap(u)$ at every program point $u$ that conservatively
describes all states that can arise at program point $u$.

The map $\prf$ constitutes a proof of sequential correctness, as
required by our algorithm, and the predicate map $\predmap$ is a
valid basis for the proof. The map $\predmap$ can be extended into
a positive basis for the proof easily enough.
Since a minimal proof
can lead to better
concurrency control, approaches that produce a ``parsimonious proof''
(\eg, see~\cite{BLAST:POPL04}) are preferable.
A parsimonious proof is one that avoids the use of unnecessary predicates
at any program point.

\subsection{Complete Schema}
\label{sec:complete-schema}

We now present a complete outline of our schema for synthesizing concurrency
control.

\begin{enumerate}[(1)]
\item
   Construct a sequential proof $\prf$ that the library satisfies the given
   assertions in any sequential execution.
\item
   Construct positive basis $\predmap$ and an obligation mapping $\obliged$
   for the proof $\prf$.
\item
   Compute a map $\mayFalsify$ from the edges of the control graph to
   $\predset$, the range of $\predmap$, such that $\mayFalsify(e)$ (conservatively)
   includes all
   predicates in $\predset$ that may be falsified by the execution of $e$.
\item
   Compute the equivalence relation $\lockeq$ on $\predset$.
\item
   Generate a predicate lock allocation map $\lockmap:\predset \rightarrow \lockset$
   such that for any $\pred_1 \lockeq \pred_2$, we have
   $\lockmap(\pred_1) = \lockmap(\pred_2)$.
\item
  Compute the following quantities for every edge $e = \stmtedge{u}{s}{v}$,
  where we use $\lockmap(X)$ as shorthand for $\Set{ \lockmap(p) \vbar p \in
  X}$ and $\bothmap(u) = \predmap(u) \cup \obliged(u)$:
   \[
   \begin{array}{lll}
   \PLocksAcquired(e) & = &
      \lockmap(\bothmap(v)) \setminus \lockmap(\bothmap(u)) \\
   \PLocksReleased(e) & = &
      \lockmap(\bothmap(u)) \setminus \lockmap(\bothmap(v)) \\
   \BreakLocks(e) & = &
      \lockmap(\mayFalsify(e)) \setminus \lockmap(\bothmap(u)) \setminus
      \lockmap(\bothmap(v))
   \end{array}
   \]
\item
   We obtain the concurrency-safe library $\widehat{\lib}$ by transforming
   every edge $\stmtedge{u}{s}{v}$ in the library $\lib$ as follows:
   \begin{enumerate}
  \item
	$\forall~p~\in~\PLocksAcquired(\stmtedge{u}{s}{v})$,
     add an {\acqpred{$\lockmap(p)$}} before $s$; 
  \item
	$\forall~p~\in~\PLocksReleased(\stmtedge{u}{s}{v})$,
     add a {\relpred{$\lockmap(p)$}} after $s$;
  \item
     $\forall~p~\in\BreakLocks(\stmtedge{u}{s}{v})$,
     add an {\acqpred{$\lockmap(p)$}} before $s$ and a {\relpred{$\lockmap(p)$}} after $s$.
   \end{enumerate}
   All lock acquisitions along a given edge are added in an order consistent
   with a total order established on all locks.
\end{enumerate}

\subsection{Correctness}
\label{sec:correctness}

We now present a formal statement of the correctness claims for our
algorithm.
Let $\lib$ be a given library with a set of embedded assertions
satisfied by all sequential executions of $\lib$.
Let $\widehat{\lib}$ be the library obtained by augmenting $\lib$
with concurrency control using the algorithm presented in
Section~\ref{sec:complete-schema}.
Let $\prf$, $\predmap$, and $\obliged$ be the proof, the positive basis, and
the obligation map used to generate $\widehat{\lib}$.

Consider any concurrent execution of the given library $\lib$.
We say that a thread $t$ is \emph{safe} in a state $\state$ if
$\concsat{\state}{t}{\prf(u)}$ where $t$'s program-counter in state
$\state$ is $u$. We say that thread $t$ is \emph{active} in state $\state$ if
its program-counter is something other than the quiescent vertex.
We say that state $\state$ is \emph{safe} if every active thread $t$
in $\state$ is safe. 
Recall that a concurrent execution is of the form:
$\state_0 \tredge{\ell_0} \state_1 \tredge{\ell_1} \cdots \state_n$, where
each label $\ell_i$ is an ordered pair $(t,e)$ indicating that the
transition is generated by the execution of edge $e$ by thread $t$.
We say that a concurrent execution is safe if every state in the
execution is safe.
It trivially follows that a safe execution satisfies all assertions
of $\lib$.

Note that every concurrent execution $\pi$ of $\widehat{\lib}$ corresponds
to an execution $\pi'$ of $\lib$ if we ignore the transitions corresponding
to lock acquire/release instructions. We say that an execution $\pi$ of
$\widehat{\lib}$ is safe if the corresponding execution $\pi'$ of $\lib$
is safe. The goal of the synthesized concurrency control is to ensure that
only safe executions of $\widehat{\lib}$ are permitted.

We say that a transition $\state \tredge{(t,e)} \state'$ preserves the
basis of an active thread $t' \neq t$ whose program-counter in state $\state$ is $u$
if for every predicate $\pred \in \predmap(u)$ the following holds:
if $\concsat{\state}{t'}{\pred}$, then $\concsat{\state'}{t'}{\pred}$.
We say that a transition $\state \tredge{(t,e)} \state'$ ensures the
basis of thread $t$ if either $e = \edge{x}{y}$ is not the procedure
entry edge or for every active thread $t' \neq t$
whose program-counter in state $\state$ is $u$
and for every predicate $\pred \in \predmap(u)$
none of the predicates in $\predmap(y)$ are in $\obliged(u)$.

We say that a transition $\state \tredge{(t,e)} \state'$ is \emph{basis-preserving}
if it preserves the basis of every active thread $t' \neq t$ and ensures
the basis of thread $t$. A concurrent execution is said to be \emph{basis-preserving}
if all transitions in the execution are basis-preserving.

\begin{lem}
\label{lemma:seqcorr}
(a) Any basis-preserving concurrent execution of $\lib$ is safe.
(b) Any concurrent execution of $\widehat{\lib}$ corresponds to
a basis-preserving execution of $\lib$.
\end{lem}

\begin{proof}
(a) We prove that every state in a basis-preserving execution of $\lib$
is safe by induction on the length of the execution.

Consider a thread $t$ in state $\state$ with program-counter value $u$.
Assume that $t$ is safe in $\state$. Thus,
$\concsat{\state}{t}{\prf(u)}$.
Note that $\prf(u)$ can be expressed in terms of the predicates in
$\predmap(u)$ using conjunction and disjunction.
Let $\textit{SP}$ denote the set of all predicates $\pred$ in $\predmap(u)$
such that
$\concsat{\state}{t}{\pred}$.
Let $\state'$ be any state such that
$\concsat{\state'}{t}{\pred}$ for every $\pred \in \textit{SP}$.
Then, it follows that $t$ is safe in $\state'$. Thus, it follows that
after any basis-preserving transition every thread that was
safe before the transition continues to be safe after the transition.

We now just need to verify that whenever an inactive thread becomes
active (representing a new procedure invocation), it starts off being
safe. We can establish this by inductively showing that every library invariant
must be satisfied in a given state or must be in $\obliged(u)$ for some
active thread $t$ at vertex $u$.

(b) Consider a concurrent execution of $\widehat{\lib}$. We need to show
that every transition in this execution, ignoring lock acquires/releases,
is basis-preserving. This follows directly from our locking scheme.
Consider a transition $\state \tredge{(t,e)} \state'$.
Let $t' \neq t$ be an active thread whose program-counter in state $\state$ is $u$.
For every predicate $\pred \in \predmap(u) \cup \obliged(u)$, our scheme ensures that
$t'$ holds the lock corresponding to $\pred$.
As a result, both the conditions for preserving basies are satisfied.
\end{proof}

\begin{thm}
\label{theorem:seqcorr}
(a) Any concurrent execution of $\widehat{\lib}$ satisfies every assertion of
$\lib$.
(b) The library $\widehat{\lib}$ is deadlock-free.
\end{thm}

\begin{proof}

(a) This follows immediately from Lemma~\ref{lemma:seqcorr}.

(b) This follows from our scheme for merging locks and can be proved by
contradiction. Assume that a concurrent execution of $\widehat{\lib}$ produces
a deadlock. Then, we must have a set of threads $t_1$ to $t_k$ and a set of
locks $\ell_1$ to $\ell_k$ such that each $t_i$ holds lock $\ell_i$ and is
waiting to acquire lock $\ell_{i \oplus 1}$, where $i \oplus 1$ denotes
 $(i~\mathrm{mod}~k) + 1$. In particular,
$t_i$ must hold lock $\ell_i$ because it wants a lock on some predicate $p_i$,
and must be trying to acquire lock $\ell_{i \oplus 1}$ because of
some predicate $q_{i \oplus 1}$. Thus, we must have
$q_i \lockeq p_i$ and $p_i \mustpreceedarrow q_{i \oplus 1}$ for every $i$.
This implies that all of $p_i$ and $q_i$ must be in the same equivalence class
of $\lockeq$ and, hence, $\ell_1$ through $\ell_k$ must be the same, which is
a contradiction (since we must have $k > 1$ to have a deadlock).
\end{proof}

As mentioned earlier, our synthesis technique has a close connection to
Owicki-Gries~\cite{owicki76verifying} approach to verifying concurrent programs.
An alternative approach to proving Theorem~\ref{theorem:seqcorr}(a) would be to construct
a suitable Owicki-Gries style proof for the library.
We believe that this is doable.

\newcommand{\schedule}{\xi}
\newcommand{\instrvar}[1]{\ensuremath{\tilde{#1}}}
\newcommand{\std}[1]{\underline{#1}}
\newcommand{\instrsat}[2]{{#1} \models_i {#2}}
\newcommand{\writemap}{\mathfrak{mf}}
\newcommand{\LPstmt}{\ensuremath{\mathcal{LP}}}

\section{Handling 2-State Assertions}
\label{sec:twostate}

The algorithm presented in the previous section can be extended to handle 
2-state assertions via a simple program transformation that allows us to treat 
2-state assertions (in the original program) as single-state assertions
(in the transformed program). We augment the set of local variables with
a new variable $\instrvar{v}$ for every (local or shared)
variable $v$ in the original program and add a primitive statement
$\LPstmt{}$ at the entry of every procedure, whose execution essentially
copies the value of every original variable $v$ to the corresponding new
variable $\instrvar{v}$.

Let $\std{\state'}$ denote the projection of a transformed program
state $\state'$ to a state of the original program obtained by
forgetting the values of the new variables.
Given a 2-state assertion $\spec$, let $\instrvar{\spec}$ denote the
single-state assertion obtained by replacing every $v^{in}$ by $\instrvar{v}$.
As formalized by the claim below, the satisfaction of a 2-state assertion
$\spec$ by executions in the original program corresponds to satisfaction
of the single-state assertion $\instrvar{\spec}$ in the transformed program.

\begin{lem}\hfill
\begin{enumerate}[\em(1)]
\item
A schedule $\schedule$ is feasible in the transformed program
iff it is feasible in the original program.

\item
Let $\state'$ and $\state$ be the states produced by a particular
schedule with the transformed and original programs, respectively.
Then, $\state = \std{\state'}$.

\item
Let $\exec'$ and $\exec$ be the executions produced by a particular
schedule with the transformed and original program, respectively.
Then,
$\exec$ satisfies a single-state assertion $\assertion$ iff
$\exec'$ satisfies it.
Furthermore,
$\exec$ satisfies a 2-state assertion $\spec$ iff
$\exec'$ satisfies the corresponding one-state assertion $\instrvar{\spec}$.
\end{enumerate}
\end{lem}

\paragraph{Synthesizing concurrency control.}
We now apply the technique discussed in Section~\ref{sec:pres_assn}
to the transformed program to synthesize concurrency control that preserves
the assertions transformed as discussed above.
It follows from the above Lemma that this concurrency control, used with
the original program, preserves both single-state and two-state assertions.

\section{Implementation}

We have built a prototype implementation of our algorithm.
Our implementation takes a sequential library and its assertions as input.
It uses a pre-processing phase to combine the library with a harness
that simulates the execution of any possible sequence of library calls
to get a complete C program.
(This program corresponds to the control graph described in Section~\ref{sec:problem}.)
It then uses a verification tool to generate a proof of correctness for the
assertions in this program.
We use the predicate-abstraction based software verification tool Yogi
described in \cite{Yogi:TSE} to generate the required proofs.
We modified the verifier to emit the proof from the final abstraction,
which associates every program point with a boolean formula over predicates.
It then uses the algorithm presented in this paper to synthesize concurrency
control for the library.
It utilizes the theorem prover Z3~\cite{Z3:TACAS} to identify the statements
in the program whose execution may falsify relevant predicates.

\setlength{\tabcolsep}{2pt}
\begin{table}[!t]
\centering
\small{
  \begin{tabular}{|l|p{4in}|r|}
  \hline
  \textbf{Library} & \textbf{Description}\\\hline
  \hline
  $compute.c$  
  & See \Fref{fig:motivating-example}
  \\\hline
  $reduce.c$ 
  & See \Fref{fig:lzty-sep-example}
  \\\hline
  $increment.c$ 
  & See \Fref{fig:control-flow-example}
  \\\hline
  $average.c$
  & Two procedures that compute the running sum and average of a sequence of numbers
  \\\hline
  $device\_cache.c$
  & One procedure that reads data from a device and caches the data
  for subsequent reads~\cite{shaz09calculus}. The specification
  requires quantified predicates.
  \\\hline
  $server\_store.c$
  & A library derived from a Java implementation of Simple
  Authentication and Security Layer (SASL). The library stores
  security context objects for sessions on the server side.
  \\\hline
\end{tabular}
}
\caption{Benchmarks used in our evaluation.}
\label{tab:benchmarks} 
\end{table}

We used a set of benchmark programs to evaluate our approach (Table ~\ref{tab:benchmarks}).
We also applied our technique manually to two real world libraries,
a device cache library ~\cite{shaz09calculus},
and a C implementation of the Simple Authentication and Security Layer (SASL).
The proofs for the device cache library and the SASL library
require quantified predicates, which were beyond the scope of the verifier we used.

In all these programs, the concurrency control scheme we synthesized
was identical to what an experienced programmer would generate.
The concurrency control we synthesized required one lock for all libraries,
with the exception of the SASL library, where our solution uses two locks.
Our solutions permit more concurrency as compared to a naive solution that uses one
global lock or an atomic section around the body of each procedure.
For example, in case of the server store library, our
scheme generates smaller critical sections and identifies a larger
number of critical sections that acquire different locks as compared
to the default implementation.
For these examples, the running time of our approach is
dominated by the time required to generate the proof; the time
required for the synthesis algorithm was negligible.

The source code for all our examples and their concurrent versions are
available online at~\cite{wypiwyg-examples}.
Note that our evaluation studies only small programs. 
We leave a more detailed evaluation of our approach as future work.

\newcommand{\ievt}{\textit{inv}}
\newcommand{\revt}{\textit{r}}
\newcommand{\RetVal}{\ensuremath{\mathit{ret}}}
\newcommand{\RPstmt}{\ensuremath{\mathcal{RP}}}
\newcommand{\retpred}{\phi}
\newcommand{\retvarp}{\ensuremath{\mathit{ret}'}}
\newcommand{\valmap}{\theta}

\newcommand{\sched}{\zeta}

\section{Concurrency Control For Linearizability}
\label{sec:linearizability}

\subsection{The Problem}

In the previous section, we showed how to derive concurrency
control to ensure that each procedure satisfies its sequential
specification even in a concurrent execution. 
However, this may still be too permissive,
allowing interleaved executions that produce counter-intuitive results
and preventing compositional reasoning in clients of the library. 
E.g., consider the procedure {\tt Increment} shown in
Fig.~\ref{fig:increment}, which increments a shared variable {\tt x}
by {\tt 1}.  The figure shows the concurrency control derived using
our approach to ensure specification correctness.  Now consider a
multi-threaded client that initializes {\tt x} to {\tt 0} and invokes
{\tt Increment} concurrently in two threads.  It would be natural to
expect that the value of {\tt x} would be {\tt 2} at the end of any
execution of this client.  However, this implementation permits an
interleaving in which the value of {\tt x} at the end of the execution
is {\tt 1}: the problem is that both invocations of {\tt Increment}
individually meet their specifications, but the cumulative effect
is unexpected\footnote{
We conjecture that such concerns do not arise when the specification
does not refer to global variables. For instance, the specification
for our example in Fig.~\ref{fig:motivating-example} does not refer
to global variables, even though the implementation uses global
variables.}.

\begin{figure}[!t]
\centering
\lstset{language=C}
\lstset{xleftmargin=15pt}
\lstset{tabsize=2}
\lstset{float=t}
\lstset{basicstyle={\ttfamily\footnotesize}}
\lstset{numbers=left, stepnumber=1,numbersep=4pt}
\lstset{emph={acquire,release},emphstyle=\color{blue}, emph={[2]Increment}, emphstyle={[2]\color{red}}}
\begin{lstlisting}[mathescape]
int x = 0;
$//@ensures~x == x^{in} + 1 \wedge returns~x$
Increment() {
	int tmp;
	acquire($l_{(x==x^{in})}$); tmp = x; release($l_{(x==x^{in})}$);
	tmp = tmp + 1;
	acquire($l_{(x==x^{in})}$); x = tmp; release($l_{(x==x^{in})}$);	
	return tmp;
}	
\end{lstlisting}
\caption{\label{fig:increment} 
A non-linearizable implementation of the procedure {\tt Increment}}
\end{figure}

This is one of the difficulties with using pre/post-condition specifications
to reason about concurrent executions.

One solution to this problem is to apply concurrency control synthesis
to the code (library) that calls {\tt Increment}.
The synthesis can then detect the potential for interference between the
two calls to {\tt Increment} and prevent them from happening concurrently.
Another possible solution, which we explore in this section, is for the
library to guarantee a stronger correctness criteria called
\emph{linearizability}~\cite{HerlihyWing}. \lzty\ gives the illusion
that in any concurrent execution, (the sequential specification of)
every procedure of the library appears to execute
\emph{instantaneously} at some point between its call and return.
This illusion allows clients to reason about the behavior of
concurrent library compositionally using its sequential specifications.
In this section, we show how our approach presented earlier for synthesizing
a \emph{logical concurrency control} can be adapted to derive
concurrency control mechanisms that guarantee linearizability.

\subsubsection{Linearizability}

We now extend the earlier notation to define linearizability.
Linearizability is a property of the library's externally observed behavior.
A library's interaction with its clients can be described in
terms of a \emph{history}, which is a sequence of events, where each event
is an \emph{invocation} event or a \emph{response} event.
An invocation event is a tuple consisting of the procedure invoked,
the input parameter values for the invocation, as well as a unique identifier.
A response event consists of the identifier of a preceding invocation event,
as well as a return value. Furthermore, an invocation event can have
at most one matching response event. A complete history has
a matching response event for every invocation event.
Note that an execution, as defined in Section~\ref{sec:problem}, captures
the internal execution steps performed during a procedure execution.
A history is an abstraction of an execution that captures only procedure
invocation and return steps.

A sequential history is an alternating sequence $\ievt_1, \revt_1, \cdots, \ievt_n, \revt_n$ of
invocation events and corresponding response events.  We abuse our earlier
notation and use $\state+\ievt_i$ to denote an entry state corresponding to a
procedure invocation consisting of a valuation $\state$ for the library's
global variables and a valuation $\ievt_i$ for the invoked procedure's formal
parameters.  We similarly use $\state+\revt_i$ to denote a procedure exit
state with return value $\revt_i$.  Let $\state_0$ denote the value of the
globals in the library's initial state. Let $\spec_i$ denote the specification
of the procedure invoked by $\ievt_i$. 
A sequential history is \emph{legal} if there exist valuations $\sigma_i$, $1
\leq i \leq n$, for the library's globals such that
$\seqsat{(\state_{i-1}+\ievt_i,\state_i+\revt_i)}{\spec_i}$ for $1 \leq i \leq
n$.

A complete interleaved history $H$ is {\em linearizable} if
there exists some legal sequential history $S$ such that
(a) $H$ and $S$ have the same set of invocation and response events and
(b) for every return event $\revt$ that precedes an invocation event $\ievt$
in $H$, $\revt$ and $\ievt$ appear in that order in $S$ as well.
An incomplete history $H$ is said to be linearizable if the complete
history $H'$ obtained by appending some response events and omitting some
invocation events without a matching response event is linearizable.

Finally, a library $\lib$ is said to be linearizable if every 
history produced by $\lib$ is linearizable. 

The concept of a \emph{linearization point} is often used in explanations
and proofs of correctness of linearizable algorithms. Informally, a linearization
point is a point (or control-flow edge) inside the procedure such that the
procedure appears to execute atomically when it executes that point. Our
eventual goal is to parameterize our synthesis algorithm with a linearization
point specification (a description of the point or points we wish to serve as
the linearization point). In this paper, however, we treat the procedure entry
edge as the linearization point and will refer to it as the linearization point.

\subsubsection{Implementation As A Specification and Logical Serializability}

The techniques we present in this section actually guarantee linearizability
with respect to the given sequential implementation (\ie, treating the sequential
implementation as a sequential specification). 
In particular, this approach guarantees that the concurrent execution will
return the same values as some sequential execution.
(The word \emph{atomicity} is sometimes used to describe this behavior.)
Such an approach has both advantages as well as disadvantages.
The advantage is that the technique is more broadly applicable, in practice,
as it does not require a user-provided specification. The disadvantage is that,
in theory, the sequential implementation may be more restrictive than the
intended specification. Hence, preserving the sequential implementation behavior
may unnecessarily restrict concurrency.

The properties of atomicity and linearizability relate to the externally
observed behavior of the library (\ie, the behavior as seen by clients of the
library). The implementation technique we use also guarantees certain 
properties about the internal (execution) behavior of the library, which we
explain now.

Recall that an execution, as defined in Section~\ref{sec:problem}, captures
the internal execution steps performed during a procedure execution
while a history is an abstraction of an execution that captures only procedure
invocation and return steps.

Recall that every transition $\sigma \longstrOf{(t,e)} \sigma'$ is labelled
by a pair $(t,e)$, indicating that the transition was created by the
execution of edge $e$ by thread $t$. We refer to a pair of the form
$(t,e)$ as a \emph{step}. A schedule $\sched$ is a sequence of steps
$\ell_1, \cdots, \ell_k$. We say that a schedule $\ell_1, \cdots, \ell_k$
is \emph{feasible} if there exists an execution
$\sigma_0 \ctrOf{\ell_1} \sigma_1 \cdots \ctrOf{\ell_k} \sigma_k$, where
$\sigma_0$ is the initial program state.
Given an execution $\pi$ = $\sigma_0 \ctrOf{\ell_1} \sigma_1 \cdots \ctrOf{\ell_k} \sigma_k$,
the sub-schedule of $t$ in $\pi$  is the sequence $\ell_{s_1}, \cdots, \ell_{s_n}$
of steps executed by $t$ in $\pi$.

A procedure invocation $t_1$ is said to precede another procedure
invocation $t_2$ in an execution if $t_1$ completes before $t_2$
begins.

Two complete executions are said to be observationally-equivalent if
they consist of the same set of procedure invocations and for each
procedure invocation the return values are the same in both executions.
An execution $\pi_1$ is said to be a permutation of another execution $\pi_2$ if for
every thread (procedure invocation) $t$ the sub-schedule of $t$ in $\pi_1$
and $\pi_2$ are the same.
An execution $\pi_1$ is said to be topologically consistent with another
execution $\pi_2$ if for every pair of procedure invocations $t_1$
and $t_2$, if $t_1$ precedes $t_2$ in $\pi_1$ then $t_1$ precedes $t_2$
in $\pi_2$ as well.

Our goal is to synthesize a concurrency control mechanism that permits
only executions that are observationally-equivalent, topologically
consistent, permutations of sequential executions.
We note that this concept is similar to various notions of serializability~\cite{BOOK:WV01}
(commonly used in database transactions).
The new variant we exploit may be thought of as \emph{logical} serializability:
corresponding points in the compared executions satisfy equivalence with respect
to certain predicates of interest, as determined by the basis. 

\subsubsection{Terminology}

In this section, we will use a modified notion of basis introduced in
Section~\ref{sec:pres_assn}. 

The key idea we explore in this paper is that of precisely characterizing
what is \emph{relevant} to a thread at a particular point and using this
information to derive a concurrency control solution. In the previous sections,
we captured the relevant information as an invariant or set of predicates
(the basis). In this section, we will find it necessary to mark certain values
(\eg, the value of a variable at a particular program point) as relevant as
well. In order to seamlessly reason about such relevant values (\eg, of type integer)
along with relevant predicates, we utilize \emph{symbolic} predicates to encode the
relevance of values.

In the sequel, note that we consider two predicates to be equal only if
they are syntactically equal.

\label{sec:symbolic}

A symbolic predicate is one that utilizes auxiliary (logical) variables.
As an example, given program variable {\tt x} and a logical variable
$w$, we will make use of predicates such as ``{\tt x} = $w$''. Such
symbolic predicates can be manipulated just like normal predicates
(e.g., in computing weakest-precondition). Conceptually, such a symbolic
predicate can be interpreted as a short-hand notation for the (possibly
infinite) family of predicates obtained by replacing the logical variable
$w$ by every possible value it can take. Thus, if {\tt x} and $w$ are of
type $T$, then the above symbolic predicate represents the set of predicates
$\{ \texttt{x} = c \vbar c \in T \}$. Note that this set of predicates
captures the value of {\tt x}: \ie, we know the value of every predicate in
this set iff we know the value of {\tt x}. This trick lets us use the
symbolic predicate ``{\tt x} = $w$'' to indicate that the value of {\tt x}
is relevant to a thread (and, hence, should not be modified by another thread).

Given any predicate $\pred$, let $\pred^*$ denote the set of predicates
it represents (obtained by instantiating the logical variables in $\pred$ as
explained above). (Thus, for a non-symbolic predicate $\pred$, $\pred^* = \{\pred\}$.)
We say that $\pred_1$ and $\pred_2$ are equivalent if $\pred_1^* = \pred_2^*$.
E.g., if $w$ ranges over all integers, then ``{\tt x} = $w$'' and  ``{\tt x} = $w+1$''
are equivalent predicates. Predicate equivalence can be used to simplify a
set of predicates or a basis. Given a set of predicates $S$, let $S^*$ represent
the set of predicates $\cup \{ \pred^* \vbar \pred \in S \}$. If $S_1^* = S_2^*$,
then it is safe, in the sequel, to replace the set $S_1$ by the set $S_2$
in a basis. This may be critical in creating finite representations of certain basis.

We say that a predicate $\pred$ is covered
by a set of predicates $S$ if $\pred$ can be expressed as a boolean formula
over the predicates in $S$ using conjunctions and disjunctions.

Recall that a \emph{predicate mapping} is a
mapping $\predmap$ from the vertices of the control graph to a set of predicates.

We say that a predicate mapping $\predmap$ is \emph{wp-closed} if for
every edge $e = \stmtedge{u}{s}{v}$ and for every $\pred \in \predmap(v)$,
(a) If $e$ is not the entry edge of a procedure, then the weakest-precondition
of $\pred$ with respect to $s$, $\mathit{wp(s,\pred)}$ is covered by $\predmap(u)$, and
(b) If $e$ is the edge $\edge{\qvertex}{\entryVertex{\Proc}}$ from the quiescent
vertex to the entry vertex of $\Proc$, then $\pred'$ is covered by $\predmap(\qvertex)$,
where $\pred'$ is obtained by replacing the occurrence of any procedure parameter
$x_i$ by a new logical variable $x_i'$.

Finally, we say that a predicate mapping is \emph{closed} if it is wp-closed and
if for every vertex $u$ and every predicate $\pred$ in $\predmap(u)$,
the negation of $\pred$ is also in $\predmap(u)$. The later condition helps us
reuse the algorithm description from Section~\ref{sec:pres_assn} in spite of some
differences in the context. 

Without loss of generality, we assume that each procedure $P_j$
returns the value of a special local variable $ret_j$.

\subsection{The Synthesis Algorithm}

We now show how our approach can be extended to guarantee linearizability
or atomicity.
We use a few tricky cases to motivate the adaptations we use of our
previous algorithm.

We start by characterizing non-linearizable interleavings permitted by
our earlier approach. We classify the interleavings based on the
nature of linearizability violations they cause.  For each class of
interleavings, we describe an extension to our approach to generate
additional concurrency control to prohibit these interleavings.
Finally, we prove correctness of our approach by showing that all
interleavings we permit are linearizable.

\subsubsection{Delayed Falsification} \label{sec:all-at-once}
The first issue we address, as well as the solution we adopt, are
not surprising from a conventional perspective. (This extension is,
in fact, the analogue of two-phase locking: \ie, the trick of
acquiring all locks before releasing any locks to avoid interference.)
Informally, the problem with the \texttt{Increment} example can
be characterized as ``dirty reads'' and ``lost updates'': the
second procedure invocation executes its linearization point later
than the first procedure invocation but reads the original value of
{\tt x}, instead of the value produced by the the first invocation.
Dually, the update done by the first procedure invocation is lost,
when the second procedure invocation updates {\tt x}.
From a logical perspective, the second invocation relies on the
invariant $x==x^{in}$ early on, and the first invocation 
breaks this invariant later on when it assigns to {\tt x}
(at a point when the second invocation no longer relies on the
invariant).
This prevents us from reordering the execution to construct an
equivalent sequential execution (while preserving the proof).
To achieve linearizability, we need to avoid such ``delayed
falsification''.

The extension we now describe prevents such interference by
ensuring that instructions that may falsify predicates and
occur after the linearization point appear to execute atomically 
at the linearization point. We achieve this by modifying the strategy to
acquire locks as follows.
\begin{iteMize}{$\bullet$}
\item We generalize the earlier notion of \emph{may-falsify}.  We say that a
  path \emph{may-falsify} a predicate $\pred$ if some edge in the path
  may-falsify $\pred$.  We say that a predicate $\pred$
  \emph{may-be-falsified-after} vertex $u$ if there exists some path from $u$
  to the exit vertex of the procedure that does not contain any linearization
  point and may-falsify $\pred$.

\item Let $\writemap$ be a predicate map such that for any vertex $u$,
$\writemap(u)$ includes any predicate that may-be-falsified-after $u$.

\item We generalize the original scheme for acquiring locks.
We augment every edge $e = \stmtedge{u}{S}{v}$ as follows:
\begin{enumerate}[(1)]
\item 
$\forall~\ell~\in~\lockmap(\writemap(v)) \backslash
\lockmap(\writemap(u))$, add an ``{\acqpred{$\ell$}}''
before $S$
\item 
$\forall~\ell\in~\lockmap(\writemap(u)) \backslash
\lockmap(\writemap(v))$, add an ``{\relpred{$\ell$}}''
after $S$
\end{enumerate}
\end{iteMize}\vspace{3 pt}

\noindent This extension suffices to produce a linearizable implementation
of the example in Fig.~\ref{fig:increment}.

\newcommand{\red}[1]{\textcolor{red}{#1}}
\addtolength{\subfigbottomskip}{-5pt}
\addtolength{\subfigtopskip}{-5pt}
\addtolength{\subfigcapmargin}{-5pt}
\begin{figure}[!t]
\centering
\begin{minipage}[h]{5.1in}
    \hspace{-2em}
    \subfigure[][]{%
        \label{fig:rvi-ex-a}
        \begin{minipage}[b]{0.1in}
        {\footnotesize {\tt
        \begin{ntabbing}
\reset
1\=12\=12\=123\=123\=123\=123\=123\=123\=123\=123\=123\kill
\kw{int} x, y;                                     \\ 
\kw{IncX}() \{                                     \\ 
\>\kwb{acquire($l_{x==x^{in}}$);}                \\ 
\>x = x + 1;                                       \\ 
\>($ret_{11}$,$ret_{12}$)=(x,y);                   \\ 
\>\kwb{release($l_{x==x^{in}}$);}                \\ 
\}                                                 \\ 
\reset                                                
\kw{IncY}() \{                                     \\ 
\> \kwb{acquire($l_{y==y^{in}}$);}               \\ 
\> y = y + 1;                                      \\ 
\> ($ret_{21}$,$ret_{22}$)=(x,y);                  \\ 
\> \kwb{release($l_{y==y^{in}}$);}               \\ 
\}                                                 \\
        \end{ntabbing}
        }}
        \end{minipage}
    }%
    \hspace{-1em}
    \subfigure[][]{%
        \label{fig:rvi-ex-b}
        \begin{minipage}[b]{0.1in}
        {\footnotesize {\tt 
        \begin{ntabbing}
1\=12\=12\=123\=123\=123\=123\=123\=123\=123\=123\=123\kill
\kw{int} x, y;                                                             \\
{\scriptsize $@ensures~x=x^{in}+1$} \\
{\scriptsize $@returns~(x, y)$}                                \\
\kw{IncX}() \{                                                             \\
\red{[$ret'_{11}$==x+1$~\wedge$~$ret'_{12}$==y]}                             \\
\> $\LPstmt{}:$ x = x$^{in}$                                               \\
\red{[x==x$^{in}$ $\wedge$ $ret'_{11}$==x+1~$\wedge$ $ret'_{12}$=y]}                                             \\
\> x = x + 1;                                                              \\
\red{[x==x$^{in}$+1 $\wedge$ $ret'_{11}$==x~$\wedge$ $ret'_{12}$= y]}                                            \\
\> ($ret_{11}$,$ret_{12}$)=(x,y);                                          \\
\red{[x==x$^{in}$+1 $\wedge$ $ret_{11}$==$ret'_{11}$} \\
\red{$\wedge$ $ret_{12}$==$ret'_{12}$]}                                   \\
\}                                                                         \\
        \end{ntabbing}
        }}
        \end{minipage}
    }%
    \hspace{-2em}
    \subfigure[][]{%
        \label{fig:rvi-ex-c}
        \begin{minipage}[b]{0.1in}
        {\footnotesize {\tt
        \begin{ntabbing} 
1\=12\=12\=123\=123\=123\=123\=123\=123\=123\=123\=123\kill
\kw{int} x, y;                                         \\
\kw{IncX}() \{                                         \\
\> \kwb{acquire($l_{\textit{merged}}$);}               \\
\> x = x+1;                                            \\
\> ($ret_{11}$,$ret_{12}$)=(x,y);                      \\
\> \kwb{release($l_{\textit{merged}}$);}               \\
\}                                                     \\
\kw{IncY}() \{                                         \\
\> \kwb{acquire($l_{\textit{merged}}$);}               \\
\> y = y+1;                                            \\
\> ($ret_{21}$,$ret_{22}$)=(x,y);                      \\
\> \kwb{release($l_{\textit{merged}}$);}               \\
\}                                                     \\
        \end{ntabbing}
        }}
        \end{minipage}
    }%

\caption[Return Value Interference]{%
\label{fig:lzty-sep-example}
An example illustrating return value interference.  Both procedures return
{\tt (x,y)}.  $ret_{ij}$ refers to the $j^{th}$ return variable of the
$i^{th}$ procedure. Figure~\ref{fig:rvi-ex-a} is a non-linearizable implementation synthesized using the approach described in Section~\ref{sec:pres_assn}. Figure~\ref{fig:rvi-ex-b} shows the extended proof of correctness of the procedure {\tt IncX} and Figure~\ref{fig:rvi-ex-c} shows the linearizable implementation.}
\end{minipage}
\end{figure}

\subsubsection{Return Value Interference}
We now focus on interference that can affect \emph{the actual value returned
by a procedure invocation}, leading to non-linearizable executions.

Consider  procedures {\tt IncX} and {\tt IncY} in Fig.~\ref{fig:lzty-sep-example},
which increment variables {\tt x} and {\tt y} respectively.
Both procedures return the values of {\tt x} \emph{and} {\tt y}.
However, the postconditions of {\tt IncX} (and {\tt IncY}) do not
\emph{specify anything about the final value of} {\tt y} (and {\tt x} respectively).
Let us assume that the linearization points of the procedures are their entry points.
Initially, we have $x = y = 0$.
Consider the following interleaving of a concurrent execution of the
two procedures.
The two procedures execute the increments in some order, producing the
state with $x = y = 1$. Then, both procedures return $(1,1)$.
This execution is non-linearizable because in any legal sequential
execution, the procedure executing second is obliged to return a value
that differs from the value returned by the procedure executing
first.
The left column in \Fref{fig:lzty-sep-example} shows the concurrency
control derived using our approach with previously described extensions.
This is insufficient to prevent the above interleaving.  This
interference is allowed because the specification for {\tt IncX} allows it to
change the value of {\tt y} arbitrarily; hence, a concurrent modification to
{\tt y} by any other procedure is not seen as a hindrance to {\tt IncX}.

To prohibit such interferences within our framework, we need to determine
whether the execution of a statement $s$ can potentially affect the return-value
of another procedure invocation. We do this by computing a predicate
$\retpred(\retvarp)$ at every program point $u$ that captures the relation
between the program  state at point $u$ and the value returned by the procedure
invocation eventually (denoted by $\retvarp$). We then check if the execution of
a statement $s$ will break predicate $\retpred(\retvarp)$, treating $\retvarp$
as a free variable, to determine if the statement could affect the return value of
some other procedure invocation.

Formally, we assume that each procedure returns the value of a special variable
$\RetVal$. (Thus, ``\texttt{return} $exp$'' is shorthand for ``$\RetVal =
exp$''.) We introduce a special auxiliary variable $\retvarp$. We say that
a predicate map $\predmap$ \emph{covers return statements} if for every
edge $\edge{u}{v}$ labelled by a return statement ``\texttt{return} $exp$''
the set $\predmap(u)$ covers the predicate $\texttt{{\retvarp} == \RetVal}$.
(See the earlier discussion in Section~\ref{sec:symbolic} about such symbolic
predicates and how they encode the requirement that the value of {\tt \RetVal}
at a return statement is relevant and must be preserved.)

By applying our concurrency-control synthesis algorithm to a closed basis
that covers return statements, we can ensure that no return-value interference
occurs.

The middle column in \Fref{fig:lzty-sep-example} shows the augmented sequential
proof of correctness of {\tt IncX}.
The concurrency control derived using our approach starting with this
proof is shown in the third column of Fig.~\ref{fig:lzty-sep-example}. The lock
$l_\textit{merged}$ denotes a lock obtained by merging locks
corresponding to multiple predicates simultaneously acquired/released.
It is easy to see that this
implementation is linearizable. Also note that if the shared variables
{\tt y} and {\tt x} were \emph{not} returned by procedures {\tt IncX} and
{\tt IncY} respectively, we will derive a locking scheme in which
accesses to {\tt x} and {\tt y} are protected by different locks,
allowing these procedures to execute concurrently. 

\subsubsection{Control Flow Interference}
An interesting aspect of our scheme is that it permits interference
that alters the control flow of a procedure invocation if it does
not cause the invocation to violate its specification.
Consider procedures {\tt ReduceX} and {\tt IncY} shown in
Fig.~\ref{fig:control-flow-example}.
The specification of {\tt ReduceX} is that it will produce a final state where
$x < y$, while the specification of {\tt IncY} is that it will increment
the value of $y$ by 1. {\tt ReduceX} meets its specification by setting
$x$ to be $y-1$, but does so \emph{only if} $x \geq y$.

\begin{figure}[!t]
\centering
\begin{multicols}{2}
\lstset{language=C}
\lstset{xleftmargin=15pt}
\lstset{tabsize=2}
\lstset{float=t}
\lstset{basicstyle={\ttfamily\footnotesize}}
\lstset{numbers=left, stepnumber=1,numbersep=4pt}
\lstset{emph={acquire,release},emphstyle=\color{blue}, emph={[2]ReduceX, IncY}, emphstyle={[2]\color{red}}}
\begin{lstlisting}[mathescape]
int x, y;
$//@ensures~y = y^{in} + 1$
IncY() {
	$[true]~\LPstmt{}:$ y$^{in}$ = y
	$[y == y^{in}]$ y = y + 1;
	$[y == y^{in} + 1]$ 
}
\end{lstlisting}
\columnbreak
\begin{lstlisting}[mathescape]
$//@ensures~x < y$
ReduceX() {
	$[true]~\LPstmt{}$
	$[true]$ if (x $\ge$ y) {
	$[true]$ 	x = y - 1;
			 }
	$[x < y]$ 
}
\end{lstlisting}
\end{multicols}
\caption{\label{fig:control-flow-example} An example illustrating
interference in control flow. Each line is annotated (in square braces)
with a predicate the holds at that program point.}
\end{figure}

Now consider a client that invokes {\tt ReduceX} and {\tt IncY} concurrently
from a state where $x = y = 0$.
Assume that the {\tt ReduceX} invocation enters the procedure.
Then, the invocation of {\tt IncY} executes
completely. The {\tt ReduceX} invocation continues, and does nothing since
$x < y$ at this point.

\Fref{fig:control-flow-example} shows a sequential proof and the
concurrency control derived by the scheme so far, assuming that the
linearization points are at the procedure entry.  A key point to note is that
{\tt ReduceX}'s proof needs only the single predicate $x < y$. The statement $y
= y + 1$ in {\tt IncY} does \emph{not falsify} the predicate $x < y$; hence,
{\tt IncY} does not acquire the lock for this predicate.
This locking scheme permits {\tt IncY} to execute concurrently with
{\tt ReduceX} and affect its control flow.  
While our approach guarantees that this control flow interference will
not cause assertion violations, proving linearizability in the presence
of such control flow interference, in the general case, is challenging
(and an open problem).

We now describe how our technique can be extended to prevent control
flow interference,
which suffices to guarantee linearizability.

We ensure that interference by one thread does not affect the execution
path another thread takes. We say that a basis $\predmap$ \emph{covers
the branch conditions} of the program if for every branch edge $\stmtedge{u}{s}{v}$,
the set $\predmap(u)$ covers the assume condition in $s$.
If we synthesize concurrency control using a \emph{closed} basis
$\predmap$ that covers the branch conditions, we can ensure that no
control-flow interference happens.

In the current example, this requires predicate $x \ge y$ to be
added to the basis for {\tt ReduceX}. As a result, {\tt ReduceX} will
acquire lock $l_{x \ge y}$ at entry, while {\tt IncY} will acquire the
same lock at its linearization point and release the lock after the
statement $y = y + 1$. It is easy to see that this
implementation is linearizable.

\subsubsection{The Complete Schema}

In summary, our schema for synthesizing concurrency control that
guarantees linearizability is as follows.

First, we determine a closed basis for the program that covers
all return statements and branch conditions in the program.
(Such a basis is the analogue of the proof and basis used in
Section~\ref{sec:pres_assn}. An algorithm for generating such a basis
is beyond the scope of this paper. Such a basis can be computed
by iteratively computing weakest-preconditions, but, in the general
case, subsumption and equivalence among predicates will need to
be utilized to simplify basis sets to ensure termination.)
We then apply the extended concurrency control synthesis algorithm
described in Section~\ref{sec:all-at-once}.

\subsection{Correctness}

The extensions described above to the algorithm of Sections~\ref{sec:pres_assn}
and~\ref{sec:twostate} for synthesizing concurrency control are sufficient to
guarantee linearizability, as we show in this section.

\newcommand{\truePred}[2]{\mathrm{TP}({#1},{#2})}
\newcommand{\trueBasisPred}[2]{\mathrm{TBP}({#1},{#2})}
\newcommand{\lschedule}{\zeta}

Let $\sigma$ be a program state.
We define $\trueBasisPred{\sigma}{t}$ to be the set
$\{ \pred \in (\predmap(u))^* \vbar \concsat{\sigma}{t}{\pred} \}$
where $t$'s program-counter in state $\sigma$ is $u$.
(See Section~\ref{sec:symbolic} for the definition of $S^*$ for any
set of predicates $S$.)

\begin{lem}
\label{lemma:monotone}
Let $\predmap$ be a wp-closed predicate map. Consider transitions
$\sigma_1 \longctrOf{(t,e)} \sigma_2$ and
$\sigma_3 \longctrOf{(t,e)} \sigma_4$.
If $\trueBasisPred{\sigma_1}{t} \supseteq \trueBasisPred{\sigma_3}{t}$,
then $\trueBasisPred{\sigma_2}{t} \supseteq \trueBasisPred{\sigma_4}{t}$.
\end{lem}

\begin{proof}
Let $e$ be the edge $\stmtedge{u}{S}{v}$. Note that for every predicate
$\pred \in \predmap(v)$, the weakest-precondition of $\pred$ with
respect to the statement $S$ can be expressed in terms of the predicates
in $\predmap(u)$ using conjunction and disjunction (by definition of a
wp-closed predicate map). The result follows.
\end{proof}

Consider any concurrent execution $\pi_1$ produced by a schedule $\schedule$.
We assume, without loss of generality, that every procedure invocation
is executed by a distinct thread. Let $t_1, \dots, t_k$ denote the
set of threads which complete execution in the given schedule, ordered
so that $t_i$ executes its linearization point before $t_{i+1}$.  
We show that $\schedule$ is linearizable by showing that $\schedule$
is equivalent to a sequential execution of the specifications of the
threads $t_1, \dots, t_k$ executed in that order.

Let $\schedule_i$ denote a {\em projection} of schedule $\schedule$
consisting only of execution steps by thread $t_i$. Let $\lschedule$
denote the schedule $\schedule_1 \cdots \schedule_k$.

\begin{lem}
\label{lemma:correspondence}
$\lschedule_k$ is a feasible schedule. Furthermore, for any
corresponding execution steps $\sigma_j \longctrOf{(t,e)} \sigma_{j+1}$
and $\sigma'_k \longctrOf{(t,e)} \sigma'_{k+1}$ of the two executions, we
have $\trueBasisPred{\sigma_j}{t} \supseteq \trueBasisPred{\sigma'_k}{t}$.
\end{lem}

\begin{proof}
Proof by induction over the execution steps of $\lschedule$.

The claim is trivially true for the first step of $\lschedule$, since
the initial state in the same in both executions.

Now, consider any pair of ``candidate'' successive execution steps
$\sigma'_{k-1} \longctrOf{(t,e')} \sigma'_k \longctrOf{(t,e)} \sigma'_{k+1}$
of $\lschedule$. That is, we assume, from our inductive hypothesis,
that the first execution step above is feasible, but we need to establish
that the second step is a feasible execution step.

Let $\sigma_{m-1} \longctrOf{(t,e')} \sigma_m$ and
$\sigma_j \longctrOf{(t,e)} \sigma_{j+1}$ be the two corresponding
execution steps in the original execution.

Our inductive hypothesis guarantees that
$\trueBasisPred{\sigma_m}{t} \supseteq \trueBasisPred{\sigma'_k}{t}$.
But any concurrent execution is guaranteed to be interference-free.
Hence, it follows that
$\trueBasisPred{\sigma_j}{t} \supseteq \trueBasisPred{\sigma_m}{t}$.
Hence, it follows that
$\trueBasisPred{\sigma_j}{t} \supseteq \trueBasisPred{\sigma'_k}{t}$.

Now, if $e$ is a conditional branch statement labelled with the statement
``assume $\pred$'', then we must have $\concsat{\sigma_j}{t}{\pred}$.
It follows that $\concsat{\sigma'_k}{t}{\pred}$. (This follows because
we use a basis that covers all branch conditions.) Thus, the second
candidate execution step of $\lschedule$ is indeed a feasible execution
step.

It then follows from Lemma~\ref{lemma:monotone} that 
$\trueBasisPred{\sigma_{j+1}}{t} \supseteq \trueBasisPred{\sigma'_{k+1}}{t}$.

Now, consider any pair of successive execution steps
$\sigma'_{k-1} \longctrOf{(t_{h-1},e')} \sigma'_k \longctrOf{(t_h,e)} \sigma'_{k+1}$
of $\lschedule$. Thus, we consider the first step executed by thread $t_h$
after thread $t_{h-1}$ executes its last step.

Note that
$\trueBasisPred{\sigma'_{k}}{t_{h-1}} = \trueBasisPred{\sigma'_{k}}{t_h}$
(since none of the basis predicates at the quiescent vertex involve thread-local
variables).  

Let $\sigma_p \longctrOf{(t_{h-1},e)} \sigma_{p+1}$ denote the corresponding, last,
execution step performed by $t_{h-1}$ in the interleaved execution.
Let $\sigma_j \longctrOf{(t_h,e)} \sigma_{j+1}$ denote the corresponding, first,
execution step performed by $t_h$ in the interleaved execution.
By the inductive hypothesis, 
$\trueBasisPred{\sigma_{p}}{t_{h-1}} \supseteq \trueBasisPred{\sigma'_{k}}{t_{h-1}}$.

Note that in the interleaved execution $p$ may be less than or greater than $j$:
$t_{h-1}$ may or may not have completed execution by the time $t_h$ performs its
first execution step. Yet, we can establish that
$\trueBasisPred{\sigma_{j}}{t_h} \supseteq \trueBasisPred{\sigma'_{k}}{t_{h-1}}$.
This is because no thread can execute a step that will change the value of any
predicate in $\trueBasisPred{\sigma_{p}}{t_{h-1}}$ between the last step of $t_{h-1}$
and the first step of $t_h$ (no matter how these two steps are ordered during execution).

\end{proof}

\begin{lem}
\label{lemma:retval}
For $t \in \{ t_1, \cdots, t_k \}$, the value returned by procedure invocation $t$ in
$\pi_1$ is the same as the value returned by $t$ in the sequential execution
$\pi_2$ corresponding to schedule $\lschedule$.
\end{lem}

\begin{proof}
Let $\sigma_j \longctrOf{(t,e)} \sigma_{j+1}$ and $\sigma'_k \longctrOf{(t,e)} \sigma'_{k+1}$
denote the execution of the return statements by $t$ in $\pi_1$ and $\pi_2$
respectively. it follows from Lemma~\ref{lemma:correspondence} that
$\trueBasisPred{\sigma_j}{t} \supseteq \trueBasisPred{\sigma'_k}{t}$.
Suppose that $t$ returns a value $c$ in the sequential execution $\pi_2$.
Note that we use a basis that covers all return statements. Hence,
the predicate $\texttt{c == \RetVal}$ must be in $\trueBasisPred{\sigma'_k}{t}$.
It follows that $\texttt{c == \RetVal}$ must be in $\trueBasisPred{\sigma_j}{t}$
as well. Hence, $t$ returns $c$ in $\pi_1$ as well.
\end{proof}

\begin{thm}
\label{theorem:lin-corr}
Given a library $\lib$ that is totally correct with respect to a given
sequential specification, the library $\widehat{\lib}$ generated by our
algorithm is linearizable with respect to the given specification.
\end{thm}
\begin{proof}
Follows immediately from Lemma~\ref{lemma:retval}.
\end{proof}

The above theorem requires \emph{total correctness} of the library in the
sequential setting. \emph{E.g.}, consider a procedure $P$ with a specification
\texttt{ensures x==0}. An implementation that sets \texttt{x} to be 1, and
then enters an infinite loop is \emph{partially} correct with respect to
this specification (but not totally correct). In a concurrent setting, this
can lead to non-linearizable behavior, since another concurrent thread can
observe that \texttt{x} has value 1, which is not a legally observable value
\emph{after} procedure $P$ completes execution.

\subsection{Discussion}
In this section, we have presented a logical approach to synthesizing
concurrency control to ensure linearizability/atomicity. In particular,
we use predicates to describe what is \emph{relevant} to ensure correctness
(or desired properties).
Predicates enable us to describe relevance in a more fine-grained fashion,
creating opportunities for more concurrency.

We believe that this approach is promising and that there is significant
scope for improving our solution and several interesting research directions
worth pursuing.
Indeed, some basic optimizations to the scheme presented may be critical
to getting reasonable solutions.
One example is an optimization relating to frame conditions, hinted at in
Section~\ref{sec:problem}. As an example, assume that $x > 0$ is an invariant
that holds true in between procedure invocations in a sequential execution.
(Thus, this is a library invariant.) A procedure that neither reads or writes
$x$ will, nevertheless, have the invariant $x > 0$ at every program point
(to indicate that it never breaks this invariant). Our solution, as sketched,
will require the procedure to acquire a lock on this predicate and hold it
during the entire procedure. However, this is not really necessary, and can
be optimized away.
In general, the invariant or the basis at any program point may be seen as
consisting of two parts, the \emph{frame} and the \emph{footprint}. The footprint
relates to predicates that are relevant and/or may be modified by the procedure,
while the frame simply indicates predicates that are irrelevant and left untouched
by the procedure. We need to consider only the footprint in synthesizing the
concurrency control solution. We leave fleshing out the details of such optimizations
as future work.

We conjecture that the extensions presented in this section to avoid control-flow
interference is not necessary to ensure linearizability. Indeed, note that if
we can ensure that any concurrent execution is observationally equivalent and
topologically equivalent to some sequential execution, this is sufficient. Our
current technique ensures that the concurrent execution is also a permutation of
the sequential execution: \ie, every procedure invocation follows the same execution
path in both the concurrent and sequential execution.
However, our current proof of correctness relies on this property.
Relaxing this requirement is an interesting open problem.

We believe that our technique can be adapted in a straight-forward fashion
to work with linearization points other than the procedure entry (as long as the
linearization point satisfies certain constraints). Different linearization points
can potentially produce different concurrency control solutions.

We also believe that with various of these improvements, we can synthesize the
solution presented in Fig.~\ref{fig:motivating-example} as a linearizable and
atomic implementation, starting with no specification whatsoever.

\section{Related Work}\label{sec:related}

{\bf Synthesizing Concurrency Control}: 
Vechev et al.~\cite{VechevEtAl:POPL2010} present an approach for synthesizing concurrency control for a concurrent program, given a specification in the form of assertions in the program. This approach, Abstraction Guided Synthesis, generalizes the standard counterexample-guided abstraction refinement (CEGAR) approach to verification as follows. The algorithm attempts to prove that the concurrent program satisfies the desired assertions. If this fails, an interleaved execution that violates an assertion is identified. This counterexample is used to either refine the abstraction (as in CEGAR) or to restrict the program by adding some atomicity constraints. An atomicity constraint indicates that a context-switch should not occur at a given program point (thus requiring the statements immediately preceding and following the program point to be in an atomic-block) or is a disjunction of such constraints. Having thus refined either the abstraction or the program, the algorithm repeats this process.

Our work has the same high-level goal and philosophy as Vechev et al.: derive a concurrency control solution automatically from a specification of the desired correctness properties. However, there are a number of differences between the two approaches. Before we discuss these differences, it is worth noting that the concrete problem addressed by these two papers are somewhat different: while our work focuses on making a sequential library safe for concurrent clients, Vechev et al. focus on adding concurrency control to a given concurrent program to make it safe. Thus, neither technique can be directly applied to the other problem, but we can still observe the following points about the essence of these two approaches.

Both approaches are similar in exploiting verification techniques for synthesizing concurrency control.  However, our approach decouples the verification step from the synthesis step, while Vechev et al. present an integrated approach that combines both. Our verification step requires only sequential reasoning, while the Vechev et al. algorithm involves reasoning about concurrent (interleaved) executions. Specifically, we exploit the fact that a sequential proof indicates what properties are critical at different program points (for a given thread), which allows us to determine whether the execution of a particular statement (by another thread) constitutes (potentially) undesirable interference.

We present a locking-based solution to concurrency control, while Vechev et al. present the solution in terms of atomic regions. Note that if our algorithm is parameterized to use a single lock (\ie, to map every predicate to the same lock), then the generated solution is effectively one based on atomic regions.

Raza et al.~\cite{RazaEtAl:ESOP2009} present an approach for automatically parallelizing a program that makes use of a separation logic proof. This approach exploits the separation logic based proof to identify whether candidate statements for parallelization access disjoint sets of locations. Like most classical approaches to automatic parallelization, this approach too relies on a data-based notion of interference, while our approach identifies a logical notion of interference.

Several papers~\cite{flanagan05automatic,cherem08inferring,emmi07allocation,mccloskey06autolocker,Hicks06LockInference,vaziri06associating}
address the problem of
inferring lock-based synchronization for atomic sections to guarantee atomicity.
These existing lock inference schemes identify
potential conflicts between atomic sections at the granularity of data items and
acquire locks to prevent these conflicts, either all at once or using a two-phase
locking approach. Our approach is novel in using a logical notion
of interference (based on predicates), which can permit more
concurrency.

\cite{solar08sketching} describes a sketching technique to add missing
synchronization by iteratively exploring the space of candidate programs
for a given thread schedule, and pruning the search space based on
counterexample candidates.
\cite{janjua06automaticCorrecting} uses model-checking to repair errors
in a concurrent program by pruning erroneous paths from the control-flow
graph of the interleaved program execution.
\cite{vechev08inferring} is a precursor to~\cite{VechevEtAl:POPL2010},
discussed above, that considers the tradeoff between increasing parallelism
in a program and the cost of synchronization. This paper allows users to
specify limitations on what may be used as the guard of conditional critical
regions (the synchronization mechanism used in the paper), thus controlling
the costs of synchronization.
\cite{deng02invariant} allows users to specify synchronization
patterns for critical sections, which are used to infer appropriate
synchronization for each of the user-identified
region. 
\noindent
Vechev \emph{et al.}~\cite{vechev08derivingLinearizable} address the
problem of automatically deriving linearizable objects with
fine-grained concurrency, using hardware primitives to achieve
atomicity. The approach is semi-automated, and requires the developer
to provide algorithm schema and insightful manual transformations.
Our approach differs from all of these techniques in exploiting
a proof of correctness (for a sequential computation) to
synthesize concurrency control that guarantees thread-safety.

{\bf Verifying Concurrent Programs}: 
Our proposed style of reasoning is closely related to the axiomatic
approach for proving concurrent programs of Owicki \&
Gries~\cite{owicki76verifying}.
While they focus on proving a concurrent program correct,
we focus on synthesizing concurrency control.
They observe that if two statements \textit{do not interfere},
the Hoare triple for their parallel composition can be obtained from the
sequential Hoare triples. Our approach identifies statements that \textit{may
interfere} and violate the sequential Hoare triples, and then
synthesizes concurrency control to ensure that sequential assertions
are preserved by parallel composition.
%

Prior work on verifying concurrent programs~\cite{OHearn:TCS07} has also shown that
attaching invariants to resources (such as locks and semaphores) can
enable modular reasoning about concurrent programs. Our paper turns
this around: we use sequential proofs (which are modular proofs, but
valid only for sequential executions) to identify critical invariants
and create locks corresponding to such invariants and augment the
program with concurrency control that enables us to lift the
sequential proof into a valid proof for the concurrent program.\\

\section{Limitations, Extensions, and Future Work}

In this paper, we have explored the idea that proofs of correctness
for sequential computations can yield concurrency control solutions
for use when the same computations are executed concurrently.
We have adopted simple solutions in a number of dimensions in order
to focus on this central idea.
A number of interesting ideas and problems appear worth pursuing
in this regard, as explained below.

\mypara{Procedures}
The simple programming language presented in Section~\ref{sec:problem} does not
include procedures.
The presence of procedure calls within the library gives rise to a different
set of challenges. Verification tools often compute procedure summaries
to derive the overall proof of correctness. Our
approach could use summaries as proxies for procedure calls and derives
concurrency control schemes where locks are acquired and released only in the
top-level procedures. A more aggressive approach could analyze the proofs
bottom up and infer nested concurrency control schemes where locks are acquired
and released in procedures that subsume the lifetimes of the corresponding
predicates. 

\mypara{Relaxed Memory Models}
The programming language semantics we use and our proofs assume
sequential consistency. We believe it should be possible to extend the
notion of logical interference to relaxed memory models. Under a
relaxed model, reads may return more values compared to sequential
consistent executions. Therefore, we may have to consider these
additional behaviors while determining if a statement can interfere
with (the proof of) a concurrent thread. We leave this extension for
future work.

\mypara{Optimistic Concurrency Control}
Optimistic concurrency control is an alternative to pessimistic concurrency
control (such as lock-based techniques).
While we present a lock-based pessimistic concurrency control mechanism,
it would be interesting to explore the possibility of optimistic concurrency
control mechanisms that exploit a similar weaker notion of interference. 

\mypara{Choosing Good Solutions}
This paper presents a space of valid locking solutions that guarantee
the desired properties.
Specifically, the locking solution generated is dependent on several
factors: the sequential
proof used, the basis used for the proof, the mapping from basis
predicates to locks, the linearization point used, etc. Given a metric
on solutions, generating a good solution according to the given metric
is a direction for future work. E.g., one possibility is to evaluate the
performance of candidate solutions (suggested by our framework) using
a suitable test suite to choose the best one.
Integrating the concurrency control synthesis approach with the proof
generation approach, as done by~\cite{VechevEtAl:POPL2010},
can also lead to better solutions, if the proofs themselves can be
refined or altered to make the concurrency control more efficient.

\mypara{Fine-Grained Locking}
Fine-grained locking refers to locking disciplines that use an
unbounded number of locks and associate each lock with a small number
of shared objects (typically one). Programs that use fine-grained
locking often scale better because of reduced contention for locks. In
its current form, the approach presented in this paper does not derive
fine-grained locking schemes. The locking schemes we synthesize
associate locks with predicates and the number of such predicates is
statically bounded. Generalizing our approach to infer fine-grained locking from
sequential proofs of correctness remains an open and challenging
problem. 

\mypara{Lightweight Specifications}
Our technique relies on user-provided specifications for the library.
Recently, there has been interest in lightweight annotations that capture
commonly used correctness conditions in concurrent programs (such as atomicity,
determinism, and linearizability).
As we discuss in Section~\ref{sec:intro}, we believe that there is potential for
profitably applying our technique starting with such lightweight specifications
(or even no specifications).

\mypara{Class invariants}
In our approach, a thread holds a lock on a predicate from the point
the predicate is established to the point after which the predicate is
no longer used. While this approach ensures correctness, it may often
be too pessimistic. For example, it is often the case that a library
is associated with class/object invariants that characterize the
\emph{stable} state of the library's objects. Procedures in the
library may temporarily break and then re-establish the invariants at
various points during their invocation. If class invariants are known,
it may be possible to derive more efficient concurrency control
mechanisms that release locks on the class invariants at points where
the invariants are established and re-acquire these locks when the
invariants are used. Such a scheme works only when all procedures
``co-operate'' and ensure that the locks associated with the
invariants are released only when the invariant is established.

\bibliographystyle{plain}
\bibliography{wypiwyg}

\begin{thebibliography}{10}

\bibitem{wypiwyg-examples}
{WYPIWYG} examples.
\newblock http://research.microsoft.com/en-us/projects/wypiwyg/
  wypiwyg\_examples.zip, June 2009.

\bibitem{BallRajamaniSpin00}
T.~Ball and S.~K. Rajamani.
\newblock Bebop: A symbolic model checker for {B}oolean programs.
\newblock In {\em SPIN 00: SPIN Workshop}, pages 113--130. 2000.

\bibitem{Yogi:TSE}
Nels~E. Beckman, Aditya~V. Nori, Sriram~K. Rajamani, Robert~J. Simmons, SaiDeep
  Tetali, and Aditya~V. Thakur.
\newblock Proofs from tests.
\newblock {\em IEEE Trans. Software Eng.}, 36(4):495--508, 2010.

\bibitem{cherem08inferring}
Sigmund Cherem, Trishul Chilimbi, and Sumit Gulwani.
\newblock Inferring locks for atomic sections.
\newblock In {\em Proc. of PLDI}, 2008.

\bibitem{Z3:TACAS}
Leonardo~Mendon\c{c}a de~Moura and Nikolaj Bj{\o}rner.
\newblock Z3: An efficient smt solver.
\newblock In {\em TACAS}, pages 337--340, 2008.

\bibitem{deng02invariant}
Xianghua Deng, Matthew~B. Dwyer, John Hatcliff, and Masaaki Mizuno.
\newblock Invariant-based specification, synthesis, and verification of
  synchronization in concurrent programs.
\newblock In {\em Proc. of ICSE}, pages 442--452, 2002.

\bibitem{shaz09calculus}
Tyfun Elmas, Serdar Tasiran, and Shaz Qadeer.
\newblock A calculus of atomic sections.
\newblock In {\em Proc. of POPL}, 2009.

\bibitem{emmi07allocation}
Michael Emmi, Jeff Fischer, Ranjit Jhala, and Rupak Majumdar.
\newblock Lock allocation.
\newblock In {\em Proc. of POPL}, 2007.

\bibitem{flanagan05automatic}
Cormac Flanagan and Stephen~N. Freund.
\newblock Automatic synchronization correction.
\newblock In {\em Proc. of SCOOL}, 2005.

\bibitem{gulvani06synergy}
Bhargav~S. Gulavani, Thomas~A. Henzinger, Yamini Kannan, Aditya~V. Nori, and
  Sriram~K. Rajamani.
\newblock Synergy: A new algorithm for property checking.
\newblock In {\em Proc. of FSE}, November 2006.

\bibitem{henzinger02lazy}
T.~A. Henzinger, R.~Jhala, R.~Majumdar, and G.~Sutre.
\newblock Lazy abstraction.
\newblock In {\em Proc. of POPL}, pages 58--70, 2002.

\bibitem{BLAST:POPL04}
Thomas~A. Henzinger, Ranjit Jhala, Rupak Majumdar, and Kenneth~L. McMillan.
\newblock Abstractions from proofs.
\newblock In {\em Proc. of POPL}, pages 232--244, 2004.

\bibitem{HerlihyWing}
Maurice~P. Herlihy and Jeannette~M. Wing.
\newblock Linearizability: a correctness condition for concurrent objects.
\newblock {\em Proc. of ACM TOPLAS}, 12(3):463--492, 1990.

\bibitem{Hicks06LockInference}
Michael Hicks, Jeffrey~S. Foster, and Polyvios Pratikakis.
\newblock Lock inference for atomic sections.
\newblock In {\em First Workshop on Languages, Compilers, and Hardware Support
  for Transactional Computing}, 2006.

\bibitem{janjua06automaticCorrecting}
Muhammad~Umar Janjua and Alan Mycroft.
\newblock Automatic correcting transformations for safety property violations.
\newblock In {\em Proc. of Thread Verification}, pages 111--116, 2006.

\bibitem{mccloskey06autolocker}
Bill McCloskey, Feng Zhou, David Gay, and Eric~A. Brewer.
\newblock Autolocker: Synchronization inference for atomic sections.
\newblock In {\em Proc. of POPL}, 2006.

\bibitem{OHearn:TCS07}
Peter~W. O'Hearn.
\newblock Resources, concurrency, and local reasoning.
\newblock {\em Theor. Comput. Sci.}, 375(1-3):271--307, 2007.

\bibitem{owicki76verifying}
Susan Owicki and David Gries.
\newblock Verifying properties of parallel programs : An axiomatic approach.
\newblock In {\em Proc. of CACM}, 1976.

\bibitem{RazaEtAl:ESOP2009}
Mohammad Raza, Cristiano Calcagno, and Philippa Gardner.
\newblock Automatic parallelization with separation logic.
\newblock In {\em ESOP}, pages 348--362, 2009.

\bibitem{solar08sketching}
Armando Solar-Lezama, Christopher~Grant Jones, and Rastislav Bodik.
\newblock Sketching concurrent data structures.
\newblock In {\em Proc. of PLDI}, pages 136--148, 2008.

\bibitem{vaziri06associating}
Mandana Vaziri, Frank Tip, and Julian Dolby.
\newblock Associating synchronization constraints with data in an
  object-oriented language.
\newblock In {\em Proc. of POPL}, pages 334--345, 2006.

\bibitem{vechev08derivingLinearizable}
Martin Vechev and Eran Yahav.
\newblock Deriving linearizable fine-grained concurrent objects.
\newblock In {\em In Proc. of PLDI}, pages 125--135, 2008.

\bibitem{vechev08inferring}
Martin Vechev, Eran Yahav, and Greta Yorsh.
\newblock Inferring synchronization under limited observability.
\newblock In {\em Proc. of TACAS}, 2009.

\bibitem{VechevEtAl:POPL2010}
Martin~T. Vechev, Eran Yahav, and Greta Yorsh.
\newblock Abstraction-guided synthesis of synchronization.
\newblock In {\em POPL}, pages 327--338, 2010.

\bibitem{BOOK:WV01}
Gerhard Weikum and Gottfried Vossen.
\newblock {\em Transactional Information Systems: Theory, Algorithms, and the
  Practice of Concurrency Control}.
\newblock {Morgan Kaufmann}, 2001.

\end{thebibliography}

\end{document}